\documentclass[twocolumn,prx,superscriptaddress]{revtex4}
\usepackage[latin9]{inputenc}
\usepackage{color}
\usepackage{amsmath}
\usepackage{amsthm}
\usepackage{amssymb}
\usepackage{graphicx}
\usepackage{float}
\usepackage{physics}
\usepackage{algorithm}  
\usepackage{algorithmicx}  
\usepackage{algpseudocode}

\usepackage{todonotes}
\usepackage[unicode=true,
 bookmarks=true,bookmarksnumbered=false,bookmarksopen=false,
 breaklinks=false,pdfborder={0 0 1},backref=false,colorlinks=true]
 {hyperref}
\hypersetup{
 linkcolor=magenta, urlcolor=blue, citecolor=blue, pdfstartview={FitH}, hyperfootnotes=false, unicode=true}
\usepackage{cleveref}
\crefname{theorem}{Theorem}{Theorems}
\crefname{definition}{Definition}{Definitions}
\makeatletter
\@ifundefined{textcolor}{}
{%
 \definecolor{BLACK}{gray}{0}
 \definecolor{WHITE}{gray}{1}
 \definecolor{RED}{rgb}{1,0,0}
 \definecolor{GREEN}{rgb}{0,1,0}
 \definecolor{BLUE}{rgb}{0,0,1}
 \definecolor{CYAN}{cmyk}{1,0,0,0}
 \definecolor{MAGENTA}{cmyk}{0,1,0,0}
 \definecolor{YELLOW}{cmyk}{0,0,1,0}
}

\usepackage{amsfonts,amsthm}\usepackage{tabularx}\usepackage{dcolumn}\usepackage{bm}\usepackage{graphicx}\usepackage{epstopdf}
\usepackage{times}

\newtheorem{theorem}{Theorem}
\newtheorem{corollary}{Corollary}
\newtheorem{proposition}{Proposition}
\newtheorem{definition}{Definition}
\newtheorem{lemma}{Lemma}

\hypersetup{urlcolor=blue}

\makeatother

\begin{document}

\title{Enhancing Quantum Adversarial Robustness by Randomized Encodings}
\author{Weiyuan Gong}
\affiliation{Center for Quantum Information, IIIS, Tsinghua University, Beijing 100084, China}
\author{Dong Yuan}
\affiliation{Center for Quantum Information, IIIS, Tsinghua University, Beijing 100084, China}
\author{Weikang Li}
\affiliation{Center for Quantum Information, IIIS, Tsinghua University, Beijing 100084, China}
\author{Dong-Ling Deng}
\email{dldeng@tsinghua.edu.cn}
\affiliation{Center for Quantum Information, IIIS, Tsinghua University, Beijing 100084, China}
\affiliation{Shanghai Qi Zhi Institute, 41st Floor, AI Tower, No. 701 Yunjin Road, Xuhui District, Shanghai 200232, China}

\begin{abstract}
The interplay between quantum physics and machine learning gives rise to the emergent frontier of quantum machine learning, where advanced quantum learning models may outperform their classical counterparts in solving certain challenging problems. However, quantum learning systems are vulnerable to adversarial attacks: adding tiny carefully-crafted perturbations on legitimate input samples can cause misclassifications. To address this issue, we propose a general scheme to protect quantum learning systems from adversarial attacks by randomly encoding the legitimate data samples through unitary or quantum error correction encoders. In particular, we rigorously prove that both global and local random unitary encoders lead to exponentially vanishing gradients (i.e. barren plateaus) for any variational quantum circuits that aim to add adversarial perturbations, independent of the input data and the inner structures of adversarial circuits and quantum classifiers. In addition, we prove a rigorous bound on the vulnerability of quantum classifiers under local unitary adversarial attacks. We show that random black-box quantum error correction encoders can protect quantum classifiers against local adversarial noises and their robustness increases as we concatenate error correction codes. To quantify the robustness enhancement, we adapt quantum differential privacy as a measure of the prediction stability for quantum classifiers. Our results establish versatile defense strategies for quantum classifiers against adversarial perturbations, which provide 
valuable guidance to enhance the reliability and security for both near-term and future quantum learning technologies.
\end{abstract}
\maketitle

\section{Introduction}

The flourish of machine learning has led to unprecedented opportunities and achieved dramatic success in both research and commercial fields \cite{Lecun2015Deep,Jordan2015Machine}. Some notoriously challenging problems, ranging from predicting protein structures \cite{Senior2020Improved} and weather forecasting \cite{Ravuri2021Skilful} to playing the game of Go \cite{Silver2016Mastering,Silver2017Mastering}, have been cracked recently. Meanwhile, the field of quantum computation has also made tremendous progress in recent years \cite{Arute2019Quantum,Zhong2020Quantum}, giving rise to unparalleled opportunities to speedup, enhance or innovate machine learning \cite{Dunjko2018Machine,Sarma2019Machine,Amin2018Quantum,Gao2018Quantum}. 
Within this vein, ideas and concepts from the physics domain have been utilized as core ingredients for quantum machine learning algorithms \cite{Harrow2009Quantum,Lloyd2014Quantum,Lloyd2018Quantum,Hu2019Quantum,Schuld2019Quantum,Farhi2014Quantum,Peruzzo2014Variational,Mcclean2016Theory}. Notable examples in this direction include the Harrow-Hassidim-Lloyd algorithm \cite{Harrow2009Quantum}, quantum principal component analysis \cite{Lloyd2014Quantum}, quantum generative models \cite{Gao2018Quantum,Lloyd2018Quantum,Hu2019Quantum}, quantum support vector machines \cite{Schuld2019Quantum}, and variational quantum algorithms based on parametrized quantum circuits \cite{Cerezo2021Variational,Farhi2014Quantum,Peruzzo2014Variational,Mcclean2016Theory}, etc. Yet, an important issue regarding quantum learning systems concerns their reliability and security in adversarial scenarios, especially for noisy intermediate-scale quantum (NISQ) devices \cite{Preskill2018Quantum}. Here, we introduce general defense strategies by randomly encoding legitimate data samples, and analytically show their adversarial robustness in a rigorous fashion (see Fig.~\ref{fig:Illu} for illustration). 

\begin{figure}
    \centering
    \includegraphics[width=0.48\textwidth]{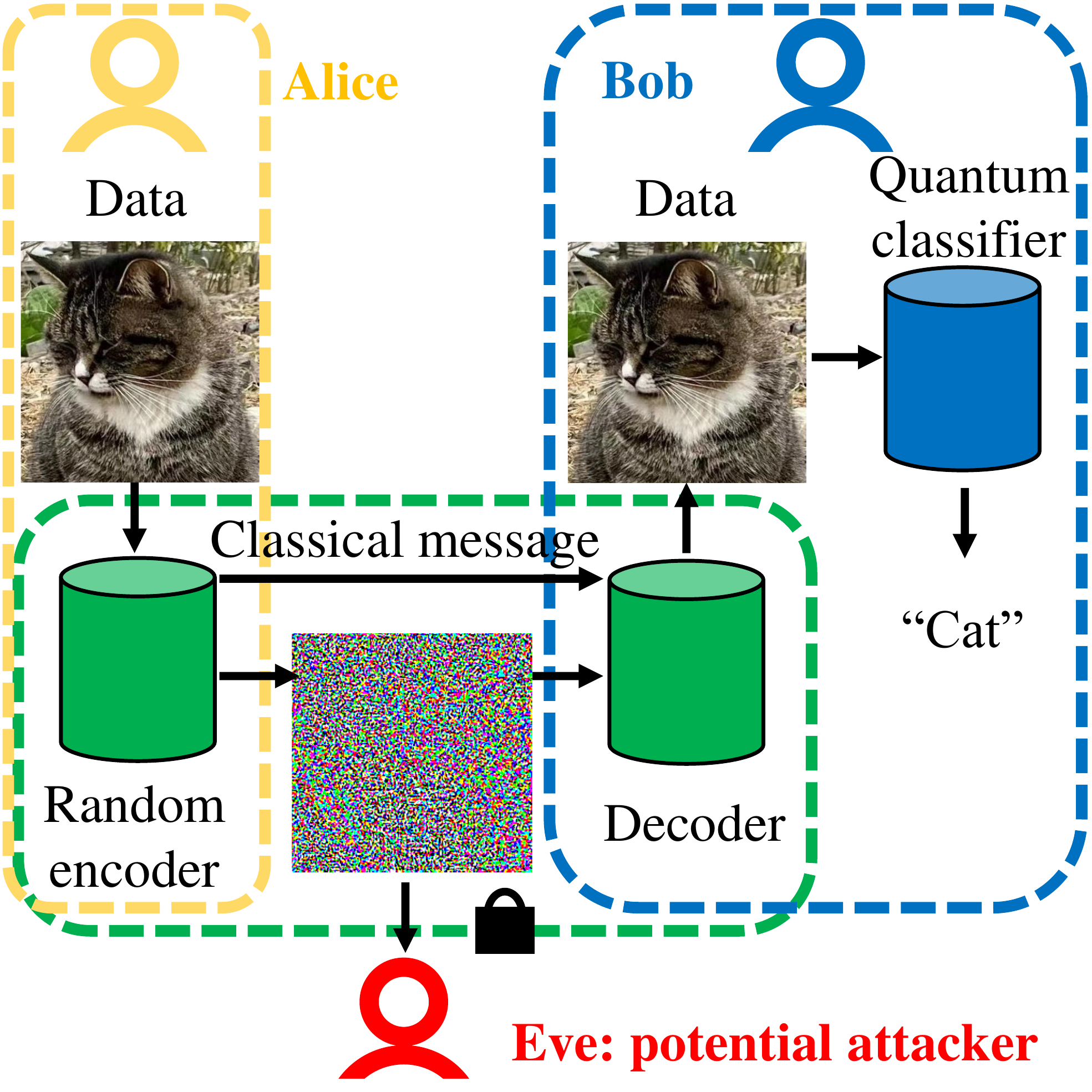}
    \caption{An illustration for exploiting randomized encoding to defend against adversarial attacks. In the quantum learning task, Alice prepares an input data sample and sends it to Bob for classification. To protect the legitimate data against the potential adversary Eve, Alice and Bob share a codebook and Alice randomly chooses an encoder in the codebook to transform the original data into encoded data, from which Eve can barely obtain any useful information. Then Alice sends to Bob the encoded quantum data and classical information about the encoder. Bob receives the messages, translates the encoded quantum data into the original figure, and performs the classification.}
    \label{fig:Illu}
\end{figure}

Adversarial machine learning is an emerging frontier that studies the vulnerability of machine learning systems and develops defense strategies against adversarial attacks \cite{Huang2011Adversarial,Chakraborty2018Adversarial}. In the classical scenario, the prediction of a deep neural network can be susceptible to tiny carefully-crafted noises, which are even imperceptible to human eyes, added to the legitimate input data \cite{Biggio2018Wild,Miller2020Adversarial,Szegedy2014Intriguing,Goodfellow2014Explaining,Kurakin2016Adversarial}. These adversarial noises can be generated by either a malicious adversary or the worst-case experimental noise from an unknown source. Recent works have demonstrated that quantum learning systems are vulnerable under adversarial settings similar to their classical counterparts \cite{Lu2020Quantum,Gong2021Universal,Liu2019Vulnerability}, sparking a new interdisciplinary research frontier of quantum adversarial machine learning \cite{Lu2020Quantum,Gong2021Universal,Ren2022Experimental,Liu2019Vulnerability,Guan2021Robustness,Liao2021Robust}.  From the theoretical aspect, even an exponentially small perturbation can cause a moderate adversarial risk for a given quantum classifier \cite{Liu2019Vulnerability}. Furthermore, it has been shown that there exist universal adversarial attacks for multiple quantum classifiers or input data samples \cite{Gong2021Universal}. 
More recently, quantum adversarial learning has been experimentally demonstrated with both large-scale real-life datasets and quantum datasets on superconducting quantum devices \cite{Ren2022Experimental}. To improve the robustness of quantum machine learning algorithms and defend against adversarial attacks, a straightforward approach is to employ a quantum-adaptive adversarial training \cite{Lu2020Quantum}. However, adversarial training in general requires generation of a large number of adversarial samples and may only perform well for the same attacking method that generates those samples. 

In classical adversarial learning, randomness is suggested to be the possible resource for developing defense strategies against adversarial perturbations \cite{Li2019Certified,Xie2018Mitigating,Guo2018Countering,Cohen2019Certified,Lecuyer2019Certified,Liu2018Towards,Pinot2019Theoretical}. However, these results are mostly empirically and there has been no unified framework for employing randomness in this context. In quantum computation, quantum error correction codes are widely used to detect and correct experimental errors. However, the errors that can be corrected are assumed to be local while adversarial perturbations are either carefully engineered or worst-case noises. In addition, the vanishing gradients (i.e. barren plateaus) for quantum circuits with randomly distributed parameters is a potential protection from most commonly used gradient-based adversarial algorithms \cite{Grigorescu2020Survey,Lu2020Quantum}.
A potential approach to achieve provable adversarial robustness for quantum classifiers is to combine randomness with quantum error correction and barren plateaus phenomenon studied in quantum computation.

In this paper, we propose an approach employing a randomized encoding procedure to protect the quantum learning systems from potential adversarial perturbations. Under practical adversarial learning scenarios, adversarial perturbations can originate from either carefully-crafted perturbations created by the attackers that have full access to the gradient information \cite{Lu2020Quantum} or the worst-case experimental noises from unknown resources \cite{Du2021Quantum}. We show the effectiveness of our scheme by using two concrete types of random encoders to mask the gradient information from the adversary and improve the robustness of quantum learning algorithms. 
The first type uses random unitary encoders and is more practical for NISQ devices, whereas the second type exploits quantum error correction encoders that are necessary for the future fault-tolerant quantum computation.

For the first type, we rigorously prove that a random global unitary encoder that satisfies $2$-design property \cite{Renes2004Symmetric} leads to exponentially small gradients for adversarial variational circuits, and thus creates barren plateaus \cite{Mcclean2018Barren,Cerezo2021Cost,Arrasmith2021Effect,Pesah2021Absence,Arrasmith2021Equivalence,Wang2021Noise,Cerezo2021Higher,Sharma2020Trainability,Holmes2021Barren,Marrero2021Entanglement,Patti2021Entanglement,Uvarov2021Barren,Grant2019Initialization,Zhao2021Analyzing,Liu2021Presence} that may hinder gradient-based algorithms in generating adversarial perturbations. We further prove that even random encoders that can be decomposed into tensor products of unitary $2$-design blocks of smaller sizes can generate barren plateaus for the adversaries as well. 
To benchmark the performance, we carry out numerical simulations concerning the classification of topological phases of the cluster-Ising model \cite{Son2011Quantum,Smacchia2011Statistical} with different loss functions and system sizes. For the second type of encoders, we consider local adversarial perturbations generated by worst-case experimental noises. We prove a lower bound for the adversarial risk in this setting based on the concentration of measure phenomenon in the high dimensional space. We analytically show that a random black-box quantum error correction (QEC) \cite{Nielsen2010Quantum,Gottesman1997Stabilizer} encoding procedure can improve the robustness of quantum learning systems for local unitary attacks. 
In particular, we show that it is sufficient to concatenate only $O(\log\log(n))$ levels of QEC encoders to bound the adversarial risk below a constant value. We adapt quantum differential privacy \cite{Hirche2022Quantum,Dwork2009Differential,Zhou2017Differential} to measure the robustness of quantum classifiers against adversarial perturbations. 
We prove an information-theoretical upper bound for the adversarial risk of quantum learning algorithms satisfying differential privacy.

The randomized encoding approach introduced in this paper is distinct from 
the previous literature that either exploit deterministic encoders for binary classification \cite{Larose2020Robust} or add white noises \cite{Du2021Quantum}. Compared to the deterministic encoder scheme which uses amplitude and phase encoding, our approach uses variational unitary circuits that are more experimental compatible for NISQ devices. Whereas adding white noise may diminish the performance of the quantum classifiers, our approach will not influence the accuracy of classification algorithms.
Furthermore, in contrast to the classical algorithms that employ randomness against adversarial attacks, our approaches provide rigorous theoretical bounds rather than empirical performance benchmarks. Our results not only establish a profound connection among quantum error correction, quantum differential privacy, barren plateau phenomenon, and quantum adversarial robustness, but also provide practical defense strategies that may prove valuable in future applications of quantum learning technologies.


The paper is organized as follows. In Sec.~\ref{sec:2}, we introduce the basic concepts and the general framework for quantum adversarial learning. In Sec.~\ref{sec:3}, we present two theorems demonstrating that both global and local randomized unitary encoders on input data samples can lead to vanishing gradients, which may hamper gradient-based algorithms from creating adversarial perturbations. We provide numerical evidence concerning classifications on the phases of the cluster-Ising model to benchmark the effectiveness of our approach. In Sec.~\ref{sec:4}, we give two theorems, one proving the vulnerability of quantum classifiers against local unitary adversarial perturbations, the other demonstrating that black-box quantum error correction encoders can effectively defend the local unitary adversarial noises on the input data samples. Finally, in Sec.~\ref{sec:5}, we discuss several open problems and conclude the paper.

\section{Basic Concepts and General Framework}\label{sec:2}
Machine learning technologies have recently achieved remarkable breakthroughs in various real-world applications \cite{Jordan2015Machine,Lecun2015Deep} including natural language processing \cite{Hinton2012Deep}, automated driving \cite{Grigorescu2020Survey}, and medical diagnostics \cite{Kononenko2001Machine}. Meanwhile, serious concerns have also been raised about the integrity and security of such technologies in various adversarial scenarios \cite{Huang2011Adversarial,Biggio2018Wild,Miller2020Adversarial}. 
For instance, the medical recognition software from a medical diagnostics or a sign recognition system from a self-driving car may cause catastrophic medical or traffic accidents if they are not robust against some occasional modifications (which may even be  imperceptible to human eyes) in identifying medical scans or traffic images \cite{Finlayson2019Adversarial}. 
To address these vital problems and concerns, the field of adversarial machine learning has been developed to construct and defend the potential adversarial manipulations against machine learning systems under different scenarios \cite{Vorobeychik2018Adversarial}. The field has attracted considerable attention and there are rapid developments for both the attack and defense strategies in different adversarial settings. For simplicity and concreteness, we will only focus our discussion on the setting of supervised learning, although 
generalizations to unsupervised or reinforcement learning settings are possible and worth systematic future investigations. 

On the one hand, there have been a number of algorithms proposed to transfer the adversarial attack problem into an optimization one and solve the corresponding problem or its variants through optimization strategies \cite{Szegedy2014Intriguing,Goodfellow2014Explaining,Biggio2018Wild,Vorobeychik2018Adversarial,Kurakin2016Adversarial,Papernot2017Practical,Madry2017Towards,Papernot2016Transferability,Papernot2016Limitations,Chen2017Zoo}. We divide the adversarial attacks into black-box and white-box attacks according to the amount of information known by the adversary about the target classifier. In the white-box setting, the attacker has  full information about the inner structure and algorithm of the classifier. Whereas, in the black-box setting the attacker possesses only partial or even no information about the classifier.   A crucial piece of information under adversarial settings is the gradient information about the classifier. The gradients can be calculated based on the inner structure, algorithm and the loss function of the classifier. In the white-box setting, various algorithms such as the fast
gradient sign method (FGSM )\cite{,Madry2017Towards}, basic iterative method (BIM) \cite{Kurakin2016Adversarial}, projected gradient descent (PGD) \cite{Madry2017Towards}, and momentum iterative method (MIM) \cite{Dong2018Boosting} have been developed based on the gradient information. In the black-box setting, algorithms that exploit   the transferability property of neural-network classifiers have been developed, including the transfer attack \cite{Goodfellow2014Explaining}, substitute model attack \cite{Papernot2017Practical,Papernot2016Transferability},
and zeroth-order optimization (ZOO) attack \cite{Chen2017Zoo} methods. On the other hand, a number of defense strategies against adversarial attacks have been developed as well. Some notable examples includes adversarial training \cite{Kurakin2017Adversarial}, defense generative adversarial network \cite{Goodfellow2014Generative,Samangouei2018Defense}, and knowledge distillation \cite{Papernot2016Distillation,Hinton2015Distilling}. These algorithms have achieved satisfying robustness performance against particular types of adversarial attacks. In general, we \textit{cannot} expect a defense strategy that can promote the robustness of all machine learning algorithms against any adversarial attacks as long as the adversary knows the information about the classifier. An alternative protocol to protect the classifier is to hide the information from the attackers. Some algorithms along this direction include adding random noise or transformations which smooths the gradients and the landscape of the loss function \cite{Li2019Certified,Cohen2019Certified,Liu2018Towards,Lecuyer2019Certified,Guo2018Countering,Xie2018Mitigating}. As a trade-off, these approaches in general would increase the difficulty in training the classifier.

Quantum classifiers are analog of classical classifiers, which aim to solve classification problems with  quantum devices \cite{Li2022Recent}.
In this paper, we propose a defense strategy for quantum classifiers against adversarial attacks through randomized encoders. We start with a brief introduction to the basic concepts, notations, and ideas of quantum classifiers and quantum adversarial learning. In general, a quantum classification task in the supervised learning setting aims to assign a label $s\in S$ to an input quantum data sample $\rho\in\mathcal{H}$, with $S$ a countable label set and $\mathcal{H}$ being a subspace of the entire Hilbert space. For technical simplicity, we suppose that the input quantum states are pure states. The supervised learning procedure aims to learn a function (called a hypothesis function) $h:\mathcal{H}\to S$ that outputs a label $s\in S$ for each input state $\rho\in \mathcal{H}$. To achieve this goal, we parametrize the hypothesis function with $\bm{\theta}\in\Xi$, where $\Xi$ is the parameter space. We train the classifier with a set of training data $\mathcal{T}_N=\{(\ket{\psi}^{(1)},s^{(1)}),...,(\ket{\psi}^{(N)},s^{(N)})\}$, where $\ket{\psi}^{(i)}$ and $s^{(i)}$ $(i=1,...,N)$ are the input states and the corresponding labels. This procedure is usually achieved by minimizing a chosen loss function $\min_{\bm{\theta}\in\Xi}L_N(\bm{\theta})$ over parameter space $\Xi$, with $L_N(\bm{\theta})=\frac{1}{N}\sum_{i=1}^NL(h(\ket{\psi}^{(i)};\bm{\theta}),s^{(i)})$ denoting the loss function averaged over the training set. A number of different quantum classifiers with different structures, loss functions, and optimization methods have been proposed \cite{Schuld2020Circuit,Farhi2018Classification,Schuld2017Implementing,Mitarai2018Quantum,Schuld2019Quantum,Havlivcek2019Supervised,Zhu2019Training,Cong2019Quantum,Wan2017Quantum,Grant2018Hierarchical,Du2021Grover,Uvarov2020Machine,Rebentrost2014Quantum,Blank2020Quantum,Tacchino2019Artificial}. Each approach bears its pros and cons, and the choice of the classifiers depends on the specific problem. A straightforward approach to construct a quantum classifier, known as variational quantum classifiers\cite{Schuld2020Circuit,Farhi2018Classification,Mitarai2018Quantum}, is to exploit variational quantum circuits \cite{Farhi2014Quantum,Peruzzo2014Variational,Mcclean2016Theory} to optimize the loss function analogously to quantum support vector machines \cite{Rebentrost2014Quantum}. There exist a number of different variants on the structures of the variational quantum circuits, including hierarchical quantum classifiers\cite{Grant2018Hierarchical} and quantum convolutional neural networks \cite{Cong2019Quantum}. 

Recent researches have shown that quantum classifiers also suffer from the vulnerability problem under adversarial attacks \cite{Lu2020Quantum,Liu2019Vulnerability,Gong2021Universal,Liao2021Robust}, with an experimental demonstration marked as the latest progress \cite{Ren2022Experimental}. 
Unlike the training procedure, finding an adversarial example for quantum classifiers can be regarded as a different optimization program on the input data space. Specifically, our goal is to discover the unitary perturbation $U_\delta$ within a restricted region $\Delta$ close to identity, which after being added to the legitimate input states, will maximize the loss function:
\begin{align}\label{eq:AdvOpt}
\max_{\delta\in\Delta} L(h(U_\delta\ket{\psi}^{(i)};\xi),s^{(i)}). 
\end{align}
In the white-box setting, the inner structures of quantum classifiers and the loss function are known to the attackers. Hence, the attackers can solve the optimization problem in Eq.~\eqref{eq:AdvOpt} exploiting the gradient information of the loss function. There have been several algorithms to attack quantum classifiers, such as quantum-adaptive BIM, FGSM, MIM algorithms \cite{Lu2020Quantum}, etc. 

\begin{figure*}
    \centering
    \includegraphics[width=0.93\textwidth]{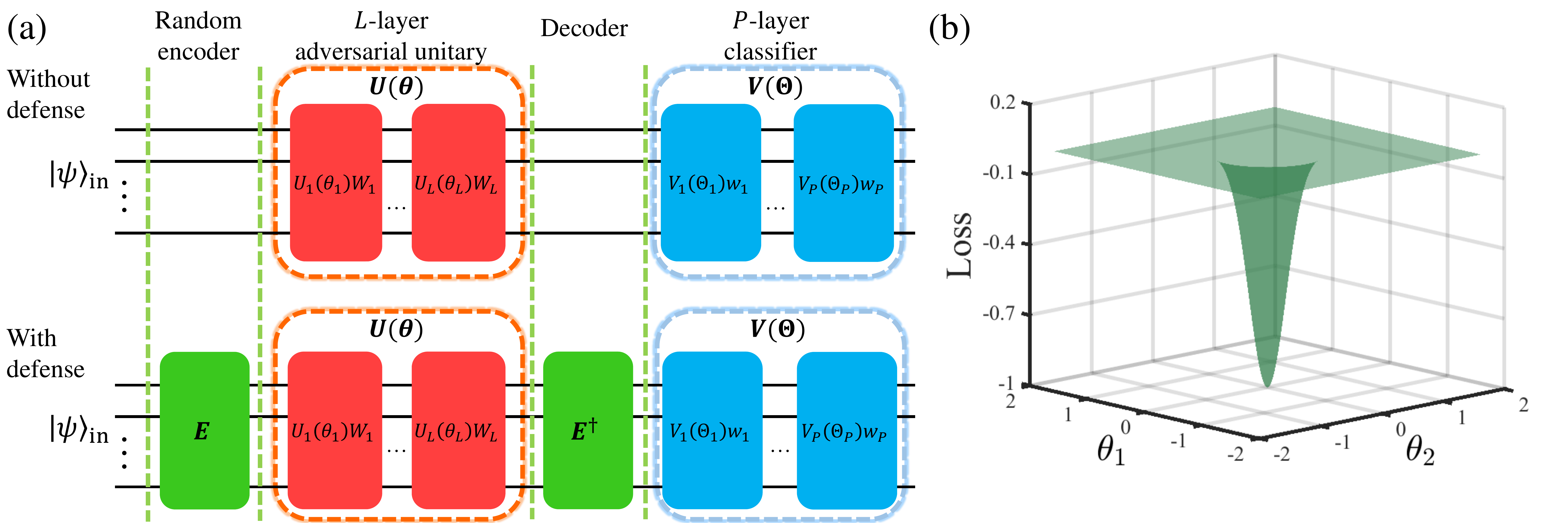}
    \caption{(a) An illustration of exploiting a random unitary encoder to defend against adversarial attack from a parametrized variational quantum circuits. In the scenario without adversarial attacks, an input state $\ket{\psi}_{\text{in}}$ is input to the parametrized variational classifier $V(\Theta)$ directly. While in the adversarial scenario, a parametrized adversarial variational circuit $U(\bm{\theta})$ is used to add an evasion attack \cite{Vorobeychik2018Adversarial}, as sketched in the upper panel. In the lower panel, a random unitary encoder $E$ and the corresponding decoder $E^\dagger$ are added before and after $U(\bm{\theta})$ to protect the data sample against potential adversary Eve. (b) By using a random global unitary encoder, the landscape for any adversarial circuit exhibits a barren plateau (i.e. vanishing gradients) regardless of the inner structure of the circuit. The variables $\theta_1$ and $\theta_2$ are variational parameters of $U(\bm{\theta})$.}
    \label{fig:BPModel}
\end{figure*}

The defense strategy under these adversarial settings remains largely unexplored, with most attention concentrated on proving the robustness of a given classifier \cite{Wiebe2018Hardening,Larose2020Robust,Weber2021Optimal}. Some notable algorithms to boost the robustness of a quantum classifier, such as adversarial training \cite{Lu2020Quantum} and adding random noise \cite{Du2021Quantum}, still suffer from white-box adversarial attacks or
the loss of useful information.

Here, we propose a generally applicable scheme to protect the quantum machine learning systems using randomized encoders in adversarial settings. Our essential idea is illustrated in Fig.~\ref{fig:Illu}.  We transfer the classification task into a three-party protocol, in which Alice prepares a legitimate quantum input data sample, Bob receives the data sample and performs the classification, and Eve is the potential adversary performing adversarial manipulations. We assume that Alice and Bob share a codebook $C=\{p_i, E_i\}$ consisting of different encoders $E_i$ with the corresponding decoders, and the probability distribution $\{p_i\}$ of choosing $E_i$. The agreement on the codebook can be realized by quantum key distribution \cite{Scarani2009Security} or quantum teleportation \cite{Nielsen2010Quantum}. We show that by randomly choosing an encoder from the codebook, the encoded quantum data can be robust against adversarial noises. Roughly speaking, the random transformation induced by the encoder masks the information that can be obtained by the adversary, thus mitigating the adversarial risk. Specifically, we consider two types of codebooks shared by Alice and Bob in the following two sections concerning the variational quantum machine learning on NISQ devices and the fault-tolerant quantum machine learning in the future. We provide analytical bounds for the robustness of protected quantum machine learning systems under adversarial settings.


\section{Defense adversarial attacks with barren plateaus}\label{sec:3}
We first consider the case of adding a random unitary transformation as an encoder. We note that 
any adversary attack can be effectively implemented as adding a $L$-layer parametrized variational quantum circuit (PVQC) $U(\bm{\theta})$, as shown in  Fig.~\ref{fig:BPModel}(a). More concretely, we can write the adversarial PVQC as 
\begin{align}\label{eq:AdvCirc}
U(\bm{\theta})=U(\theta_1,...,\theta_L)=\prod_{l=1}^LU_l(\theta_l)W_l,
\end{align}
where $U_l(\theta_l)=\exp(-i\theta_lA_l)$ is the parametrized variational component in each layer, $A_l$ is a Hermitian operator, and $W_l$ is a unitary operator that represents the fixed component in each layer. We assume the classifier $V(\Theta)$ is well-trained with parameters $\Theta$. It can be a general unitary operator such as a $P$-layer PVQC shown in the figure. To perform a prediction, we simply measure some particular qubits at the output after the classifier and assign labels according to the measurement outcomes. Given an input pure state $\ket{\psi}_{\text{in}}$, the loss function can be regarded as an expectation value over a Hermitian operator $H$. For the legitimate input and the adversarial input, the loss functions can be written as $L(\Theta)=\bra{\psi}_{\text{in}}V^\dagger(\Theta) HV(\Theta)\ket{\psi}_{\text{in}}$ and $L(\Theta;\bm{\theta})=\bra{\psi}_{\text{in}}U^\dagger(\bm{\theta})V^\dagger(\Theta) HV(\Theta)U(\bm{\theta})\ket{\psi}_{\text{in}}$, respectively.

To protect the quantum classifier $V(\Theta)$ from the adversarial PVQC $U(\bm{\theta})$, we exploit a random encoder $E$ and the corresponding decoder $E^\dagger$ to encrypt the legitimate data sample $\ket{\psi}_{\text{in}}$. We note that the codebook $C=\{E_i\}$ contains a particular set of encoders with probability distribution $p_i$. We assume that $C$ is unitary $2$-design \cite{Renes2004Symmetric}, namely that the first and the second moments are equivalent to the corresponding moments with respect to the Haar measure $d\mu_H(E)$:
\begin{align}\label{eq:TDesign}
\sum_ip_iE_i^{\otimes t}ME_i^{\dagger\otimes t}=\int d\mu_H(E)E^{\otimes t}ME^{\dagger\otimes t},t=1,2,
\end{align}
where $M$ is an arbitrary operator. As shown in Refs. \cite{Harrow2009Random,Dankert2009Exact,Brandao2016Local,Harrow2018Approximate}, quantum circuits can implement unitary $2$-design efficiently---a circuit with only $O(n^2)$ [$O(n)$] gates is sufficient for attaining exact (approximate) unitary $2$-design.  The type of the gates can be further restricted to single-qubit rotations and nearest neighbor entangling gates. Therefore, such a random encoder $E_i\in C$ can be efficiently realized by a PVQC with $O(n^2)$ gates. In this case, the loss function given a fixed encoder $E_i$ can be represented by
\begin{align}\label{eq:LossFunc}
L(\Theta,E_i;\bm{\theta})=\bra{\psi}_{\text{in}}E_i^\dagger U^\dagger E_iV^\dagger HVE_i^\dagger UE_i\ket{\psi}_{\text{in}},
\end{align}
where $U\equiv U(\bm{\theta})$ and $V\equiv V(\Theta)$ are parametrized with $\bm{\theta}$ and $\Theta$, respectively. In the adversarial setting, we assume that the adversarial PVQC is initialized with $\bm{\theta}_0$ such that $U(\bm{\theta}_0)=I$, i.e., the adversary starts from a legitimate quantum sample and explores the gradient direction to maximize the value of the loss function. We denote $\partial_{\theta_l}L(\Theta,E_i;\bm{\theta})$ to be the gradient of $L(\Theta,E_i;\bm{\theta})$ with respect to each parameter $\theta_l,l=1,...,L$ in the adversarial PVQC. Now, we are ready to present our first theorem regarding the expectation and variance on each $\partial_{\theta_l}L(\Theta,E_i;\bm{\theta})$.

\begin{theorem}\label{thm:1}
Suppose we exploit a randomly chosen global unitary encoder $E_i$ from a unitary $2$-design codebook $C=\{p_i,E_i\}$ . The expectation and variance of the derivatives of the loss function defined in Eq.~\eqref{eq:LossFunc} with respect to any component $\theta_l\in\bm{\theta}$ satisfy the following (in)equalities:
\begin{align}\label{eq:thm1-1}
&\mathbb{E}_{E_i\in C}[\partial_{\theta_l}L(\Theta,E_i;\bm{\theta}_0)]=0,\\ \label{eq:thm1-2}
&\text{Var}_{E_i\in C}[\partial_{\theta_l}L(\Theta,E_i;\bm{\theta}_0)]\leq\frac{2\Tr(A_l^2)}{d^2-1}\Tr(\rho H_V^2),
\end{align}
where $\bm{\theta}_0$ are the initial parameters for the adversarial PVQC with $U(\bm{\theta}_0)=I$, $\rho=\ket{\psi}_{\text{in}}\bra{\psi}_{\text{in}}$ is the density matrix of the input state, $A_l$ is the Hermitian operator of the parametrized variational component in the $l$-th layer, $H_V=V^\dagger HV$, and $d=2^n$ is the dimension of the Hilbert space.
\end{theorem}

\begin{proof}
We give a brief sketch of the essential idea here. The full proof is technically involved and thus left to Appendix.~\ref{appsec:1}. As we assume that the random encoder $C=\{p_i,E_i\}$ satisfies the unitary $2$-design properties and $U(\bm{\theta}_0)=I$, we can obtain the expectation and variance of th gradients by calculating the first and second moments using Haar integral. To derive the analytical results, the Haar integrals are calculated by Schur-Weyl duality\cite{Zhang2014Matrix}. We first prove that the $\mathbb{E}_{E_i\in C}[\partial_{\theta_l}L(\Theta,E_i;\bm{\theta}_0)]$ is an integral over the first moment and thus vanishes. We then calculate the variance of the gradient using $\text{Var}_{E_i\in C}[\partial_{\theta_l}L(\Theta,E_i;\bm{\theta}_0)]=\mathbb{E}_{E_i\in C}[\partial_{\theta_l}L(\Theta,E_i;\bm{\theta}_0)^2]-\mathbb{E}_{E_i\in C}[\partial_{\theta_l}L(\Theta,E_i;\bm{\theta}_0)]^2=\mathbb{E}_{E_i\in C}[\partial_{\theta_l}L(\Theta,E_i;\bm{\theta}_0)^2]$. The variance can thus be obtained by a second moment integral, which results in an exponentially small value and yields Ineq.~\eqref{eq:thm1-2}.
\end{proof}

This theorem guarantees that by choosing a random unitary encoder from the codebook $C$, we can bound the variance of gradients for any parameters in any potential adversarial PVQC circuits with an exponentially small value. By using Chebyshev's inequality, this theorem indicates that the probability of finding a gradient along any direction of amplitude larger than a fixed constant $\tau>0$ is exponentially small. It has been proved in Ref. \cite{Liu2019Vulnerability} that the vulnerability of a quantum classifier also grows exponentially with the system size $n$ and perturbations of only $O(\sqrt{1/d}),d=2^n$ can render a considerable adversarial risk. However, this result does not crack the security guarantee in our result because we prove that the gradient vanishes at a more rapid speed $O(1/d)$. The exponentially small gradients for the adversarial PVQC lead to a barren plateau which requires exponentially large precision and iteration steps for the adversary that exploits gradient-based algorithm to construct an adversarial example. Therefore, this algorithm protecting the quantum machine learning systems by masking the gradient information from the attackers. We emphasize that our protection encoder can be efficiently realized using a circuit containing only $O(n^2)$ gates to satisfy the unitary $2$-design requirement, which is roughly the same scaling as most quantum classifiers commonly used in practice.

We also stress that the adversarial PVQC is restricted to a small neighborhood of the identity operator, thus itself do not satisfy unitary 2-design. The barren plateaus faced by the adversary are induced by the random encoding process with the codebook. This is in sharp contrast to the barren plateaus for variational quantum circuits studied in the previous literature \cite{Mcclean2018Barren,Cerezo2021Cost}, where the variational circuits themselves  are required to be unitary 2-design.

\cref{thm:1} can be further extended to other codebooks. For example, we consider another model, where the encoder $E_i$ can be written as a tensor product of $m$-qubit blocks $(m<n)$ with each block satisfying unitary $2$-design. We show that using these encoders, one can create barren plateaus for the adversary PVQC with a lower request on the number of gates.  Without loss of generality, we assume that $n=m\xi$ and $E_i=\bigotimes_{j=1}^\xi E_i^j$ such that ensemble $\{p_i^j,E_i^j\}$ forms a unitary $2$-design for all $j$. We can similarly decompose the operator $A_l$ in each layer of the adversarial PVQC as:
\begin{align}\label{eq:DecompA}
A_l=\sum_k c_k \bigotimes_{j=1}^\xi A_{l,k}^j.
\end{align}

\begin{figure*}
    \centering
    \includegraphics[width=0.98\textwidth]{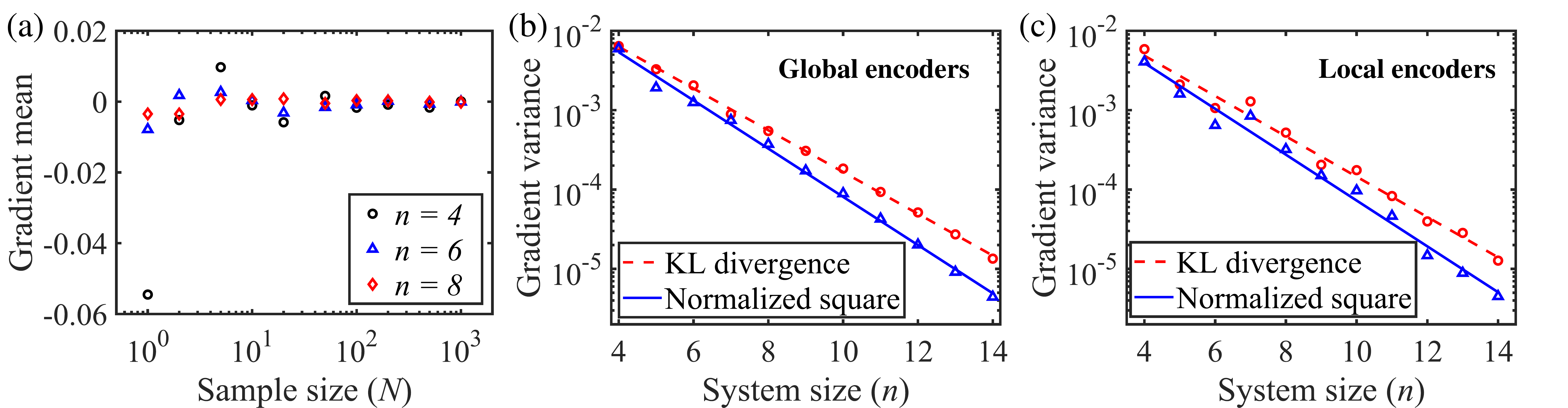}
    \caption{Numerical results for the gradients of adversarial variational circuits. (a) The mean values of $\partial_{\theta_l}L(\Theta,E_i;\bm{\theta}_0)$ as functions of the sample size $N$ for different system sizes $n$. The loss function is taken as the KL divergence. The mean values of gradients are averaged over all the parameters in the adversarial parametrized variational quantum circuits. (b) The average variances of $\partial_{\theta_l}L(\Theta,E_i;\bm{\theta}_0)$ for the KL divergence and normalized square loss as functions of the system size $n$.  The encoders used here are global random parametrized variational quantum circuits. (c) Similar to (b), whereas the encoders used here are tensor products of two-qubit random parametrized variational blocks.}
    \label{fig:Numeric}
\end{figure*}

We assume that $\sum_{k,k'}c_kc_{k'}$ is bounded, $\Tr(A_{l,k}^{j2})\leq 2^m,\forall l,i$, and $A_{l,k}^j$ is traceless. We remark that this assumption is reasonable, in the sense that it is satisfied by most commonly used quantum variational circuits. We have following theorem:

\begin{theorem}\label{thm:2}
Assume we exploit a randomly chosen encoder $E_i$,  which can be written as the tensor product of $\xi$ $m$-qubit blocks independently chosen from unitary $2$-design codebook $C=\{p_i,E_i^j\}$ $(j=1,...,\xi)$. We assume the operators in adversarial PVQC can be decomposed as Eq.~\eqref{eq:DecompA} and $A_{l,k}^j$ is traceless with $\Tr(A_{l,k}^{j2})\leq 2^m,\forall l,i,k$. The expectation and variance of the derivatives of the loss function defined in Eq.~\eqref{eq:LossFunc} with respect to any component $\theta_l\in\bm{\theta}$ satisfies the following (in)equalities:
\begin{align}\label{eq:thm2-1}
&\mathbb{E}_{E_i^j\in C}[\partial_{\theta_l}L(\Theta,E_i;\bm{\theta}_0)]=0,\\\label{eq:thm2-2}
&\text{Var}_{E_i^j\in C}[\partial_{\theta_l}L(\Theta,E_i;\bm{\theta}_0)]\leq O\left(\left(\frac{2^m+1}{2^{2m}-1}\right)^\xi\right).
\end{align} 
\end{theorem}

\begin{proof}
We sketch the main idea for the proof here and leave the technical details in Appendix.~\ref{appsec:1}. As we assume that each block $E_i^j$ in the encoder satisfies unitary $2$-design independently, we can obtain the expectation and variance of the loss function by calculating the Haar integral separately on each $E_i^j$. According to the decomposition in Eq.~\eqref{eq:DecompA}, we regard $A_l$ as a summation of terms that are tensor products of operators on each block and calculate these terms separately. In the first step, we derive the zero expectation on the gradients in Eq.~\eqref{eq:thm2-1} by calculating the first moment similar to \cref{thm:1}. Next, we compute the variance of the gradients by calculating the second moment Haar integral for each $E_i^j$. The result for the integral contains $2^\xi$ terms. We can derive the upper bound for the variance in Eq.~\eqref{eq:thm2-2} based on the assumption that $A_{l,k}^j$ is traceless with $\Tr(A_{l,k}^{j2})\leq 2^m,\forall l,i,k$.
\end{proof}

\cref{thm:2} indicates that, under particular assumptions on the adversarial PVQC, even if the encoder only satisfies unitary $2$-design on each of the subspace $SU(2^m)$ for any $m\geq 2$, the variance of the gradients for adversarial PVQC still decreases exponentially as the system size increases. By exploiting this scheme, we can reduce the gate count required in \cref{thm:1} from $O(n^2)$ to $O(\xi m^2)=O(n)$. Compared with \cref{thm:1}, the codebook requires fewer experimental resources at the price of a larger upper bound on the variance for the gradients. We mention that this encoder scheme carries over to  adversarial PVQCs with other inner structures, although we can only  analytically derive the variance bound under some constraints for the adversary due to technique difficulties.

We stress that our approach does not rely on any specific properties of the quantum classifiers $V(\Theta)$. It  does not require that $V(\Theta)$ is unitary 2-design and applies to arbitrary quantum classifiers. 
Therefore, we can avoid the barren plateau landscape when  training the quantum classifier by using shallow circuits or some quantum circuits with specific structures that are not unitary $2$-design, such as quantum convolutional neural networks \cite{Cong2019Quantum,Pesah2021Absence}. Even though we only rigorously prove the case for the loss function that can be regarded as an expectation value over Hermitian operator $H$, our method can also effectively protect quantum classifiers equipped with other loss functions. This claim is supported by the numeric results using Kullback-Leibler (KL) divergence \cite{Kullback1951Information} in the subsequent paragraphs.

To verify that the scaling results in the above theorem are valid for quantum machine learning models with modest system sizes and different loss functions, we carry out numerical simulations on classifying topological phases for the ground states of the cluster-Ising model \cite{Son2011Quantum,Smacchia2011Statistical}:
\begin{align}\label{eq:CIM}
H(\lambda)=-\sum_{i=2}^{n-1}\sigma_{i-1}^x\sigma_{i}^z
\sigma_{i+1}^x+\lambda\sum_{i=1}^n\sigma_i^y\sigma_{i+1}^y,
\end{align}
where $\sigma_i^\alpha,\alpha=x,y,z$ denotes the Pauli matrices on the $i$-th qubit and $\lambda$ is the interaction strength. Here, we take the open boundary condition. This model features a phase transition at $\lambda=1$, between the cluster phase for $0<\lambda<1$ and the antiferromagnetic phase for $\lambda>1$. We sample the Hamiltonian with a different parameter $\lambda$ from $0$ to $2$ and compute the corresponding ground states. We then construct the dataset using these ground states with the corresponding labels. We carry out the classification task using variational quantum classifiers of varying systems sizes from four to fourteen qubits and depth ten. We consider two types of loss functions for the classifiers: (i) the normalized square loss $1-\abs{\bra{\phi}\ket{\psi}_{\text{out}}}^2$ where $\ket{\psi}_{\text{out}}$ is the output state at the end of the circuit in Fig.~\ref{fig:BPModel} and $\ket{\phi}$ is the states encoded by the target labels ; (ii) the KL divergence between $\ket{\psi}_{\text{out}}$ and $\ket{\phi}$. We construct the encoder via a PVQC of four layers and sample the gradients from an adversarial PVQC of four layers. The results are obtained by averaging over variational encoders and adversarial PVQC with random parameters and input data samples. Further details for numeric results are provided in Appendix.~\ref{appsec:2}. As shown in Fig.~\ref{fig:Numeric}(a), the expectation values of the gradient along any directions in the adversarial PVQC converges to zero rapidly as we increase the number of samples, which is consistent with Eq.~\eqref{eq:thm1-1}. From Fig.~\ref{fig:Numeric}(b), we can observe that the variance of the gradients decays exponentially as the system size increase from four to fourteen qubits. The outcome from this numerical simulation fits the result for global encoder settings given by Eq.~\eqref{eq:thm1-2}. In Fig.~\ref{fig:Numeric}(c), we perform numerical experiments for the local encoder settings at $m=2$ in Eq.~\eqref{eq:thm2-2}. 
We construct the encoder by using a PVQC that can be written as a tensor product of two-qubit blocks each satisfies unitary $2$-design by randomly changes the
parameters in the block. The two-qubit blocks are set to be a two-layer variational quantum circuit with the inner structure described in Appendix~\ref{appsec:2}. We observe that the variance of gradient approaches zero rapidly as the system size increases. The numerical result 
shows the exponential decay of gradients predicted in Eq.~\eqref{eq:thm2-2}.

\section{Defending Local Adversarial Noises by Black-Box Quantum Error Correction}\label{sec:4}
\begin{figure*}
    \centering
    \includegraphics[width=0.90\textwidth]{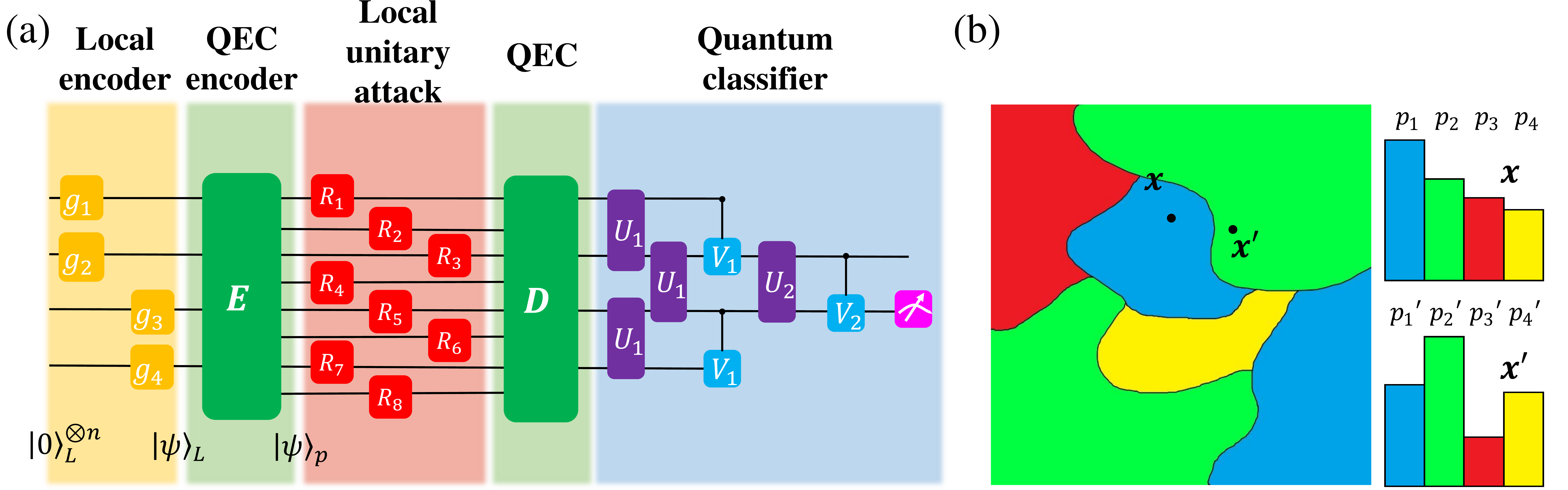}
    \caption{(a) An illustration of exploiting black-box quantum error correction (QEC) encoders to defend against local unitary adversarial attacks. An initial state $\ket{0}^{\otimes n}$ is sequentially encoded by a local encoder and a QEC encoder into the logical state $\ket{\psi}_L$ and the physical state $\ket{\psi}_p$. The physical state is exposed to local adversarial attacks from potential adversaries. It then enters a QEC decoder and is classified by a quantum classifier. (b) A sketch of the connection between quantum differential privacy (QDP) and adversarial robustness. The quantum classifier maps input data $x$ into a probability distribution $\{p_1,...,p_4\}$ and predicts its label according to the maximum likelihood. The different labels are distinguished by different colors. Given a perturbed data $x'$ with $d(x,x')\leq\tau$, $\epsilon$-QDP limits the shift on probability distribution by bounding $|\ln(p_i'/p_i)|\leq\epsilon$ $(i=1,\cdots,4)$. 
    }
    \label{fig:QECModel}
\end{figure*}

As we mentioned in the previous section, the adversarial perturbation can be regarded as experimental noises in the worst case. Under experimental settings, most operations and noises are local \cite{Arute2019Quantum,Wu2021Strong}. Therefore, in this section we consider the case in which both the adversarial perturbation and state preparation can be written as tensor products of single-qubit rotations \cite{Lu2020Quantum}. This setting is widely employed in qubit-encoding quantum computation and machine learning \cite{Giovannetti2008Quantum}. We consider the quantum classifier model $C$ mentioned in Sec.~\ref{sec:2}. We first analytically evaluate the vulnerability of quantum classifiers against such local adversarial perturbations. We suppose the quantum classifier $h:\mathcal{H}\to S$ maps the locally encoded data from $\bigotimes_{i=1}^n SU(2)$ to a label set $S=\{s_1,..,s_K\}$ that contains $K$ labels. We assume that the input data sample $g_\rho$ is chosen from $\mathcal{H}$ according to a probability measure $\mu(\cdot)$. We denote $\mu(h^{-1}(s_k))$ to be the fraction of data that will be assigned the label $s_k$ by the classifier. We now introduce the following measure of adversarial risk:

\begin{definition}\label{def:1}
Consider a hypothesis function $h:\mathcal{H}\to S$. Suppose the input data $\rho$ is chosen from $\mathcal{H}$ according to the measure $\mu(\cdot)$. 
Suppose an adversarial attack $A:\rho\to\rho',\forall \rho\in\mathcal{H}$ occurs under the constraint $d(\rho,\rho')\leq\epsilon$, we denote $M=\{\rho\in\mathcal{H}|h(\rho)\neq h(\rho')\}$ to be the set containing all the states that can be made as adversarial data samples. The adversarial risk is defined as $\mu(M)$.
\end{definition}

We consider the set of input states that can be encoded by a local unitary operator on a certain initial state (e.g., the $\ket{0}^{\otimes n}$ state) and thus the classification of the quantum data is equivalent to the classification of special unitary groups $\bigotimes_{i=1}^n SU(2)$. For technical simplicity, we assume that the input data sample $g_\rho$ is uniformly chosen from $\mathcal{H}$ according to the Haar measure for each qubit $\mu_H^{\otimes n}(\cdot)$, and denote $\mu_H^{\otimes n}(h^{-1}(s_k))$ to be the fraction of data that will be assigned the label $s_k$ by the classifier. For two states $\rho=g_\rho\ket{0}^{\otimes n}$ and $\sigma=g_\sigma\ket{0}^{\otimes n}$, where $g_\rho=\bigotimes_{i=1}^ng_\rho^i$ and $g_\sigma=\bigotimes_{i=1}^ng_\sigma^i$ are chosen from $\bigotimes_{i=1}^n SU(2)$, we exploit the normalized Hamming distance to measure the difference between $g_\rho$ and $g_\sigma$:
\begin{align}\label{eq:HamDis}
d_{\text{NH}}(g_\rho,g_\sigma)=\frac{1}{n}\sum_{i=1}^n\bm{1}[g_\rho^i\neq g_\sigma^i].
\end{align}
This normalized Hamming distance measures the fraction of unequal $g_\rho^i$ and $g_\sigma^i$ for all the qubits. We can then deduce the following theorem concerning the effectiveness of local unitary adversarial attack:

\begin{theorem}\label{thm:3}
Consider a quantum classifier that maps an input sample from $\bigotimes_{i=1}^nSU(2)$ to a $K$-label set $S=\{s_1,...,s_K\}$. Suppose we choose an operator $g_\rho$ from $\bigotimes_{i=1}^nSU(2)$ according to the Haar measure on each qubit $\mu_H^{\otimes n}(\cdot)$. Without loss of generality, we assume $\mu_H^{\otimes n}(h^{-1}(s_1))\geq\mu_H^{\otimes n}(h^{-1}(s_2))\geq...\geq\mu_H^{\otimes n}(h^{-1}(s_K))$. There always exists a perturbation $g_\rho\to g_{\rho'}$ with $d_{\text{NH}}(g_\rho,g_{\rho'})\leq\tau$, such that the adversarial risk is greater than $R\in(0,1)$ if
\begin{align}\label{eq:thm3}
\tau^2\geq\frac{1}{n}\min_{k=2,3,...,K}\ln\left[\frac{4k}{\mu_H^{\otimes n}(h^{-1}(s_k))(1-R)}\right].
\end{align}
\end{theorem}

\begin{proof}
We provide the intuition here and the technical details for the full proof are provided in Appendix.~\ref{appsec:3}. Notice that the input data space $\bigotimes_{i=1}^n SU(2)$, when equipped with the Haar measure $\mu_H^{\otimes n}(\cdot)$ and the normalized Hamming distance $d_{\text{NH}}(\cdot)$, forms a $(2,\frac{n}{2^n})$-Levy family \cite{Giordano2007Some,Talagrand1995Concentration}. By exploiting the concentration of measure phenomenon on the Levy family \cite{Mahloujifar2019Curse,Gromov1983Topological}, we show that the measure of all data within distance $\tau$ from a subset $\mathcal{H}'\subseteq \bigotimes_{i=1}^n SU(2)$ can be bounded below by $R'$, if $\tau^2\geq\frac{1}{n}\ln[\frac{4}{\mu_H^{\otimes n}(\mathcal{H}')(1-R')}]$. We then use De Morgan's law to prove that in the case of a $K$-label classification, by choosing $k=2,3,\ldots,K$ that minimizes $\ln[\frac{4k}{\mu_H^{\otimes n}(h^{-1}(s_k))(1-R)}]$, we can bound the adversarial risk below by $R$. This ends the proof of the theorem.
\end{proof}

The above theorem indicates that, for any quantum classifiers receiving input data of $n$ qubits, an adversarial attack that only changes a fraction $O(\frac{1}{\sqrt{n}})$ of qubits will result in a moderate adversarial risk bounded below by $R$. As the system size $n$ increases, the vulnerability of a quantum classifier becomes more severe even for local unitary adversarial attacks. It has been shown in Ref. \cite{Liu2019Vulnerability} that in the setting of global encoding quantum data from $SU(d)$ ($d=2^n$) and global adversarial perturbation, an perturbation of $O(\frac{1}{\sqrt{d}})$ strength under the Hilbert-Schmidt distance measure can guarantee a moderate adversarial risk. Compared with this global case, the adversarial risk under a local unitary attack is not as severe since additional constraints has been assumed for possible attacks. However, for large quantum machine learning systems, Eq.~\eqref{eq:thm3} still shows that the prediction is unstable even under a tiny noise. We remark that Eq.~\eqref{eq:thm3} still holds for other distance measures, such as the normalized Hilbert-Schmidt distance $d_{\text{NHS}}(g_\rho,g_\sigma)=\frac{1}{n}\sum_{i=1}^nd_{\text{HS}}(g_\rho^i,g_\sigma^i)$ that calculates the average Hilbert-Schmidt distance between each qubit. This follows from the fact that the normalized Hilbert-Schmidt distance is always bounded above by the normalized Hamming distance. 

The quantum error correction (QEC) codes \cite{Nielsen2010Quantum,Gottesman1997Stabilizer} are widely used in quantum computation to protect the computation from local noises and are believed to be a crucial building block in the future implementation of quantum computers. Inspired by this, it would be natural to think whether quantum error correction codes can effectively protect quantum machine learning systems from local unitary adversarial attacks. A typical QEC procedure contains an encoder $E$, an error correction circuit $\mathcal{C}$ and a decoder $D$. The encoder $E$ encodes a logical quantum state $\rho_L$ from the logical Hilbert space $\mathcal{H}_L$ into a physical quantum state $\rho_P$ from physical Hilbert space $\mathcal{H}_P$
When errors occur, we perform an error correction on physical qubits to correct particular types of errors and use the decoder to recover the original logical state $\rho_L$. A popular choice for QEC is the $[n_0,k,t]$ code \cite{Nielsen2010Quantum}, which encodes each $k$ logical qubits into $n_0$ physical qubits and is able to correct $\lfloor\frac{t}{2}\rfloor$ erroneous physical qubits for each logical qubit. Without loss of generality, we consider the case when $k=1$ and $t=3$. The corresponding QEC codes can correct one local error for each logical qubit. 

We consider the model in Fig.~\ref{fig:QECModel}(a). The adversarial settings are similar to Sec.~\ref{sec:3} except that we assume both logical state $\ket{\psi}_L$ and the adversarial attack are local. To protect the quantum machine learning systems from such adversarial attacks, we consider applying a quantum error correction encoder after the state preparation stage and the corresponding decoder before the classification stage. The QEC encoder is a black-box oracle for the adversary and thus can be regarded as a \textit{random} encoder that encodes each logical qubit into physical qubits and is able to correct particular types of local errors on these physical qubits. We remark that the assumption of a random QEC encoder can be experimentally practical. A straightforward approach is to permutate the physical qubits randomly such that the adversary does not know the corresponding encoding structures between logical qubits and physical qubits. The fact that short random circuits are good QEC codes \cite{Brown2013Short} indicates that random QEC codes can also be realized using circuits only containing $O(n\log n)$ gates.

To quantify the enhancement of adversarial robustness, we adapt the idea of quantum differential privacy (QDP) and utilize it as a measure for adversarial robustness \cite{Zhou2017Differential,Aaronson2019Gentle,Arunachalam2020Quantum}. Differential privacy \cite{Dwork2009Differential} is the property of an algorithm whose outputs can not be distinguished when inputting neighboring dataset. We can thus measure the sensitivity and vulnerability of the algorithm when changing the input using the differential privacy. For two input data samples that are separated by a small distance, differential privacy bounds the distance of the outputs after the algorithm. 
A formal definition of quantum differential privacy is given below:

\begin{definition}\label{def:2}
(Quantum differential privacy \cite{Zhou2017Differential}) Consider the quantum algorithm $\mathcal{Q}$ and a measurement $\mathcal{M}$ on its output. The algorithm $\mathcal{Q}$ is said to be $(\epsilon(\tau),\gamma)$-quantum differential privacy if for all input quantum states $\rho,\sigma$ that satisfy $d(\rho,\sigma)\leq\tau$, the following inequality holds for any possible subset $Y$ of all possible outcomes of the measurement:
\begin{align}\label{eq:def2}
\Pr[\mathcal{M}(\mathcal{Q}(\rho))\in Y]\leq e^{\epsilon(\tau)}\Pr[\mathcal{M}(\mathcal{Q}(\sigma))\in Y]+\gamma,
\end{align}
where $\epsilon(\tau)$ is a function of distance $\tau$.
\end{definition}

For technical simplicity, we focus on the case $\gamma=0$, referred to as $\epsilon$-QDP. As shown in Fig.~\ref{fig:QECModel}(b), if we consider two neighboring $x$ and $x'$ input quantum data, the $\epsilon$-QDP property bounds the difference between the probability distributions $\{p_1,p_2,p_3,p_4\}$ and $\{p_1',p_2',p_3',p_4'\}$ after the classification by bounding each $p_i/p_i'\in(e^{-\epsilon},e^{\epsilon})$. We focus on the locally encoded states in this section and exploit the normalized Hamming distance as a distance measure. It is shown in Refs. \cite{Zhou2017Differential,Du2021Quantum,Aaronson2019Gentle} that adding any amount of white noise can make the algorithm satisfy quantum differential privacy under trace distance and normalized Hamming distance. Therefore, it is reasonable to assume that under experimental settings the quantum machine learning system satisfies quantum differential privacy as we can always keep a tiny amount of quantum white noise whose influence is negligible. We remark that as distance $\tau$ increases, the corresponding $\epsilon(\tau)$ will always increase according to the definition given by Eq. \eqref{eq:def2}. For a quantum classification algorithm $\mathcal{Q}:\mathcal{H}\to\mathbb{R}^{\abs{S}}$ that maps a quantum data sample to a probability distribution on the label set $S$, the quantum differential privacy property in \cref{def:2} has direct connection with quantum adversarial risk in \cref{def:1}. The adversarial risk for $\mathcal{Q}$ can be explicitly written as
\begin{align}\label{eq:Advrisk}
\mathbb{E}_{\rho\in\mathcal{H}}\left[\sup_{\delta:\norm{\delta}\leq\tau}\Pr_{s_1\sim\mathcal{Q}(\rho),s_2\sim\mathcal{Q}(\rho+\delta)}(s_1\neq s_2)\right],
\end{align}
where $\delta$ is the adversarial perturbation and the expectation value is averaged over all choice of $\rho$ according to the measure $\mu(\cdot)$. By Jensen's inequality, the probability for two random variable chosen from probability distribution $P=(p_1,...,p_{\abs{S}})$ and $Q=(q_1,...,q_{\abs{S}})$ having the same value can be bounded by $\sum_i p_iq_i\leq\exp(\sum_ip_i\log(q_i))=\exp(-d_{KL}(P,Q)-H(P))$. Here, $H(\cdot)$ is the Shannon entropy and $d_{KL}(\cdot)$ is the KL divergence between two distributions. Noticing that the algorithm $\mathcal{Q}$ is $\epsilon(\tau)$-QDP, $P=\mathcal{Q}(\rho), Q=\mathcal{Q}(\rho+\delta)$, and $\norm{\delta}\leq\tau$, $\abs{\ln(p_i/q_i)}$ is bounded by $\epsilon(\tau)$. As a result, the KL divergence between $P$ and $Q$ is bounded above by $2\epsilon(\tau)^2$ \cite{Dwork2010Boosting}. Therefore, we can derive the following information-theoretical upper bound for the adversarial risk of $\mathcal{Q}$ that is $\epsilon(\tau)$-QDP.
\begin{proposition}\label{thm:QDPAdvrisk}
Assume we have a quantum classifier $\mathcal{Q}:\mathcal{H}\to\mathbb{R}^{\abs{S}}$ that satisfies $\epsilon(\tau)$-QDP. When performing an adversarial attack $A:\rho\to\rho'$ with $d(\rho,\rho')\leq\tau$, the adversarial risk $R(\tau)$ is bounded above by
\begin{align}\label{eq:thmQDPAdvrisk}
R(\tau)\leq 1-e^{-2\epsilon(\tau)^2}\mathbb{E}_{\rho\in\mathcal{H}}\left[e^{-H(\mathcal{Q}(\rho))}\right],
\end{align}
where the expectation value is averaged over all $\rho$ chosen uniformly from $\mathcal{H}$ according to the measure $\mu(\cdot)$.
\end{proposition}

It is worthwhile to mention that $\mathbb{E}_{\rho\in\mathcal{H}}\left[e^{-H(\mathcal{Q}(\rho))}\right]$ is a constant that only depends on the property of the classifier $\mathcal{Q}$ itself. When the classifier is well-trained and can provide the correct label with a large confidence, $\mathbb{E}_{\rho\in\mathcal{H}}\left[e^{-H(\mathcal{Q}(\rho))}\right]$ is close to $1$. In particular, if $\mathcal{Q}$ predicts the label with a unity confidence, $\mathbb{E}_{\rho\in\mathcal{H}}\left[e^{-H(\mathcal{Q}(\rho))}\right]=1$. As we decrease $\epsilon(\tau)$, the adversarial risk $R(\tau)$ will decrease polynomially. This indicates that one can improve the adversarial robustness by simply amplifying the quantum differential privacy property (i.e. decrease the $\epsilon$ parameter). This leads us to the next theorem concerning the amplification of quantum differential privacy using black-box QEC:

\begin{theorem}\label{thm:4}
Suppose we use a quantum classifier that satisfies $\epsilon(\tau)$-QDP under normalized Hamming distance. We apply a QEC encoder which encodes each logical qubit into $n_0$ physical qubits and is able to correct an arbitrary error on one of these qubits. Assume that the inner structure of the QEC is unknown by the adversary. We randomly choose arbitrary $\rho$ and $\sigma$ on the physical qubits with $d_{\text{NH}}(\rho,\sigma)\leq\tau$, then for any subset $Y$ of all possible outcomes of the measurement $\mathcal{M}$, $\Pr[\mathcal{M}(\mathcal{Q}(\rho))\in Y]\leq e^{\epsilon_{\text{QEC}}}\Pr[\mathcal{M}(\mathcal{Q}(\sigma))\in Y]$, where
\begin{align}\label{eq:thm4}
\epsilon_{\text{QEC}}=\epsilon\left(\frac{n_0(n_0-1)\tau^2}{\delta}\right),
\end{align}
with probability at least $1-\delta$.
\end{theorem}

\begin{proof}
The normalized Hamming distance measures the fraction of qubits that are different from the legitimate data. Assume an adversarial attack $\rho\to\rho'$ occurs on the physical qubits such that $d_{\text{NH}}(\rho,\rho')\leq\tau$. Under a black-box QEC procedure, a logical qubit becomes erroneous if it contains more than two erroneous physical qubits. Therefore, the expected fraction of logical qubits affected is $O(n_0(n_0-1)\tau^2)$ \cite{Nielsen2010Quantum}. As a consequence, to achieve the bound in \cref{thm:3} on the logical qubits, the adversary should alter at least $O(\frac{1}{n^{1/4}})$ fraction of physical qubits in expectation. The QEC encoder mitigates the adversarial risk in expectation.

By the Markov's inequality, one can choose $\rho$ and $\rho'$ from physical quantum states such that $d_{\text{NH}}(\rho,\rho')\leq\tau$ and
\begin{align}
d_{\text{NH}}(D\circ\mathcal{C}(\rho),D\circ\mathcal{C}(\rho'))\leq\frac{n_0(n_0-1)\tau^2}{\delta}, 
\end{align}
with probability at least $1-\delta$. By the definition of quantum differential privacy, it is guaranteed with probability at least $1-\delta$ that $\Pr[\mathcal{M}(\mathcal{Q}(\rho))\in Y]\leq e^{\epsilon_{\text{QEC}}}\Pr[\mathcal{M}(\mathcal{Q}(\rho'))\in Y]$, where $\epsilon_{\text{QEC}}=\epsilon\left(\frac{n_0(n_0-1)\tau^2}{\delta}\right)$.
\end{proof}

The above theorem indicates that for a quantum classifier satisfying quantum differential privacy, a black-box QEC encoder can effectively amplify the quantum differential privacy property with high probability. In particular, the QEC encoder can promote the $\epsilon(\tau)$-QDP quantum classifier to a new quantum classifier that satisfies $\epsilon(\frac{n_0(n_0-1)\tau^2}{\delta})$-QDP property for at least $1-\delta$ fraction of all possible input data samples. Under the normalized Hamming distance, a QEC encoder can always promote the robustness of the quantum classifier against perturbations added on the input quantum states with large probability, as long as $\tau n_0(n_0-1)/\delta\leq 1$. As $n_0$ is a constant for a fixed QEC encoder, this is a constant threshold for $\tau$.  We remark that a quantum algorithm satisfying quantum differential privacy can be obtained through adding white noise \cite{Du2021Quantum}. However, white noise will erase the information required for classification and should be suppressed for variational quantum classifiers on NISQ devices. In contrast, our approach only assumes the existence of tiny noise and can guarantee $\epsilon$-QDP for arbitrary small $\epsilon>0$ by concatenating QEC encoders. Compared with the previous work \cite{Du2021Quantum} that simply relies on the white noise to produce quantum differential privacy property, our approach prevents the risk of losing too much information due to noises.
\cref{thm:4} opens a door for studying the promotion of adversarial robustness from QEC. Intuitively, a QEC encoder can mitigate the adversarial risk from the local unitary adversarial attack by reducing the bound in Eq.~\eqref{eq:thm3} to $O(\frac{1}{n^{1/4}})$ with a large probability. This is because the QEC can reduce the error rate from $p$ to $p^2$ in expectation \cite{Nielsen2010Quantum}. We mention that it is necessary to keep the inner structure of the QEC encoder confidential to the potential adversary. If the QEC encoder is known by the attacker, the QEC circuit together with the quantum classifier can be regarded as an enlarged quantum classifier with system size $nn_0$. According to \cref{thm:3}, such a quantum classifier is more vulnerable under adversarial attacks.

In fault-tolerant quantum computation, the errors are assumed to occur locally on each qubit for each quantum operation independently with probability below the threshold $p_{\text{thres}}$. To mitigate the influence of these errors, multiple levels of QEC are concatenated to bound the error below an expected value $\zeta$. It has been proved that $O(\log\log(1/\zeta))$ levels of QEC is enough \cite{Nielsen2010Quantum}. As we mentioned in the previous sections, adversarial perturbation are not random experimental noises. Instead, these perturbations are either carefully engineered noises from hostile adversary or worst-case experimental noises. However, \cref{thm:4} indicates that we can always decrease the $\epsilon$ parameter in $\epsilon$-QDP property of the quantum classifier by concatenating additional levels of QEC. We have shown in \cref{thm:3} that adversarial perturbations with strength $O(\frac{1}{\sqrt{n}})$ can lead to moderate adversarial risk under local unitary adversarial attacks. To reduce the potential adversarial risk, we should bound the distance $\tau'=\Omega(\frac{1}{\sqrt{n}})$ after concatenating the QECs. We can thus deduce the following corollary. 
\begin{corollary}
We consider the classifier $\mathcal{Q}$ discussed in \cref{thm:4}. After concatenating $L_{\text{QEC}}$ levels of QEC, one can guarantee with high probability that $\mathcal{Q}$ satisfies $\epsilon\left(\frac{1}{\sqrt{n}}\right)$-QDP for randomly chosen $\rho$ and $\sigma$ with $d_{\text{NH}}(\rho,\sigma)\leq\tau$, as long as
\begin{align}\label{eq:coro1}
L_{\text{QEC}}\geq O(\log\log n).
\end{align}
\end{corollary}

The proof of this corollary follows from using \cref{thm:4} repeatedly and fixing the $\delta$ in each level as $\frac{\delta}{L_{\text{QEC}}}$. This corollary indicates that only $O(\log\log n)$ layers of repeated QEC encoders can guarantee $\epsilon(O(\frac{1}{\sqrt{n}}))$-QDP for the quantum classifier. The number of levels of QEC required has a double logarithmic scaling over the system size $n$. We show through this theorem that fault-tolerant quantum computers with black-box QECs are robust against adversarial attacks with a large probability. This result shows the effectiveness of QECs under the condition of even worst-case noises.

\section{Conclusions and Outlook}\label{sec:5}

In this paper, we proposed a general approach to protect quantum learning systems in adversarial scenarios using randomized encoders. We rigorously proved that random unitary encoders forming a unitary $2$-design set can create barren plateaus for any adversarial parametrized variational quantum circuit, which prevent the creations of adversarial perturbations.  To benchmark the performance of our approach, we carried out numerical simulations on classifying topological phases of ground states for the clustered-Ising Hamiltonian. We remark that this approach is feasible on NISQ devices as the the classifiers, adversarial circuits, and encoders can be implemented by variational quantum circuits. In addition, we proved that black-box quantum error correction encoders unknown to the adversary can mitigate the adversarial risk by promoting the differential privacy against local unitary noises. 
Our results develop versatile defense strategies to enhance the reliability and security of quantum learning systems, which may have far-reaching consequences in applications of quantum artificial intelligence based on both near-term and future quantum technologies.

Many questions remain and warrant further investigations. For instance, our discussions in this paper mainly focus on quantum supervised learning scenarios. Yet, unsupervised and reinforcement learning approaches may also suffer from the vulnerability problem \cite{Vorobeychik2018Adversarial}. Thus, it will be interesting and important to develop similar defense strategies in the context of quantum unsupervised or reinforcement learning, where obtaining analytical performance guarantees in a rigorous fashion might more challenging. In addition, how to extend our results to the scenario of quantum delegated learning with multiple clients \cite{Li2021Quantum} is well worth future studies. Finally, it is also of crucial importance to carry out an experiment to demonstrate our defense strategies against adversarial perturbations. This would be a key step toward secure and reliable quantum artificial intelligence technologies.

\section{Acknowledgements}

We thank L.-M. Duan, Soonwon Choi, Liwei Yu, Zidu Liu, Zhide Lu, Wenjie Jiang, and Si Jiang for helpful discussions. This work is supported by the National Natural Science Foundation of China (Grants No. 12075128 and No. T2225008) and the Shanghai Qi Zhi Institute.

\bibliographystyle{apsrev4-1-title}
\bibliography{PromAdvRobRandEncoder}
\clearpage
\onecolumngrid
\appendix
\makeatletter
\setcounter{figure}{0}
\setcounter{equation}{0}
\setcounter{table}{0}
\setcounter{definition}{0}
\setcounter{theorem}{0}
\setcounter{lemma}{0}
\setcounter{proposition}{0}
\renewcommand{\thefigure}{S\@arabic\c@figure}
\renewcommand{\thetable}{S\@arabic\c@table}
\renewcommand{\thedefinition}{S\@arabic\c@definition}
\renewcommand{\thetheorem}{S\@arabic\c@theorem}
\renewcommand{\thelemma}{S\@arabic\c@lemma}

\section{Barren Plateau Induced by Randomized Encodings}\label{appsec:1}
In this appendix, we provide the detailed proofs for \cref{thm:1} and \cref{thm:2} in the main text. We give further analytical and numerical results concerning creating barren plateaus for adversarial parametrized variational quantum circuits (PVQC). To begin with, we provide the formal definition of unitary $t$-design. Consider a polynomial $P_{t,t}(U)$ with the homogeneous degree at most $t$ in the entries of a unitary matrix $U$, and degree $t$ in the complex conjugates of these entries. We can evaluate the average of $P_{t,t}(U)$ under the Haar measure with $\sum_{i=1}^Mp_iP_{t,t}(U_i)$, where $\{U_1,...,U_M\}$ equipped with probability $\{p_1,...,p_M\}$ is called to form the unitary $t$-design. The formal definition \cite{Dankert2009Exact,Renes2004Symmetric,Mcclean2018Barren} is given as follow:

\begin{definition}
Let $P_{t,t}(U)$ be a polynomial of unitary $U$ and its complex conjugate $U^\dagger$, with up to a given degree $t$. An ensemble $\{p_i,U_i\}$, $i=1,...,M$ is called a unitary $t$-design if
\begin{align}
\sum_{i=1}^Mp_iP_{t,t}(U_i)=\int P_{t,t}(U)d\mu_{H}(U)
\end{align}
holds for any possible $P_{t,t}(U)$, where $d\mu_H(U)$ is the Haar measure.
\end{definition}

The above expression can be either exact or approximate, which corresponds to the exact unitary $t$-design or the approximate unitary $t$-design, respectively. 
The unitary $t$-design indicates that the $t$-th moments are (approximately) the same as the corresponding moments with respect to the Haar measure.
The first and second moments over the Haar measure are given by Weingarten functions \cite{Zhang2014Matrix}:
\begin{align}
&\int U^\dagger OUd\mu_H(U)=\frac{\Tr(O)}{d}I,\\
&\int U^\dagger AUXU^\dagger BUd\mu_H(U)=\frac{d\Tr(AB)-\Tr(A)\Tr(B)}{d(d^2-1)}\Tr(X)I+\frac{d\Tr(A)\Tr(B)-\Tr(AB)}{d(d^2-1)}X,
\end{align}
where $d=2^n$ is the dimension of the unitary $U$. Below we prove \cref{thm:1} in the main text.

Consider the adversarial PVQC equipped with parameters $\bm{\theta}$ and the random unitary encoder satisfying the unitary $2$-design as shown in Fig.~\ref{fig:BPModel}(a). When we fix an encoder $E_i$ and initialize the adversarial PVQC with $\bm{\theta}_0$ such that $U(\bm{\theta}_0)=I$, the gradient of the loss function $L$ can be written as:
\begin{align}\label{eq:S4}
\partial_{\theta_l}L(\Theta,E_i;\bm{\theta}_0)=i\bra{\psi}_{\text{in}}E_i^\dagger U_-^\dagger[A_l,U_+^\dagger E_iV^\dagger HVE_i^\dagger U_+]U_-E_i\ket{\psi}_{\text{in}},
\end{align}
where $U_-=\prod_{i=1}^{l-1}\text{exp}(-iA_i\theta_i)W_i$ and $U_+=\prod_{i=l}^{L}\text{exp}(-iA_i\theta_i)W_i$. At $\bm{\theta}_0$, $U_+U_-=I$. Since we have assumed the codebook $C$ equipped with the probability distribution $\{p_i,E_i\}$ forms unitary $2$-design, the average gradients at $\bm{\theta}_0$ are calculated as zero:
\begin{align}
\mathbb{E}_{E_i\in C}[\partial_{\theta_l}L(\Theta,E_i;\bm{\theta}_0)]&=i\int d\mu_H(E)\bra{\psi}_{\text{in}}E^\dagger U_-^\dagger[A_l,U_+^\dagger EV^\dagger HVE^\dagger U_+]U_-E\ket{\psi}_{\text{in}}\nonumber\\
&=i\Tr\left\{\int d\mu_H(E)\ket{\psi}_{\text{in}}\bra{\psi}_{\text{in}}E^\dagger U_-^\dagger[A_l,U_+^\dagger EV^\dagger HVE^\dagger U_+]U_-E\right\}\nonumber\\
&=\frac{i}{2^n}\Tr\left[\rho\left(\Tr\left(U_-^\dagger A_l U_+^\dagger\right)V^\dagger HV-\Tr\left(U_+A_lU_-\right)V^\dagger HV\right)\right]\nonumber\\
&=0,
\end{align}
where $\rho=\ket{\psi}_{\text{in}}\bra{\psi}_{\text{in}}$. Next, we prove the exponential decay variances of the gradients. Since $\text{Var}_{E_i\in C}[\partial_{\theta_l}L(\Theta,E_i;\bm{\theta}_0)]=\mathbb{E}_{E_i\in C}[\partial_{\theta_l}L(\Theta,E_i;\bm{\theta}_0)^2]-\mathbb{E}_{E_i\in C}[\partial_{\theta_l}L(\Theta,E_i;\bm{\theta}_0)]^2=\mathbb{E}_{E_i\in C}[\partial_{\theta_l}L(\Theta,E_i;\bm{\theta}_0)^2]$, the variances can be calculated by the second moment integral:
\begin{align}
\mathbb{E}_{E_i\in C}[\partial_{\theta_l}L(\Theta,E_i;\bm{\theta}_0)^2]&=-\Tr\left\{\int d\mu_H(E)\rho E^\dagger U_-^\dagger[A_l,U_+^\dagger EV^\dagger HVE^\dagger U_+]U_-E\rho E^\dagger U_-^\dagger[A_l,U_+^\dagger EV^\dagger HVE^\dagger U_+]U_-E\right\}\nonumber\\
&=\frac{2}{d(d^2-1)}[d\Tr(A_l^2)-\Tr^2(A_l)][\Tr(\rho H_V^2)-\Tr^2(\rho H_V)]\nonumber\\
&\leq\frac{2}{d^2-1}\Tr(A_l^2)\Tr(\rho H_V^2),
\end{align}
where $H_V=V^\dagger HV$. This completes the proof of \cref{thm:1} in the main text.

We then prove \cref{thm:2} in the main text. We utilize the following Haar measure integral over tensor-product unitary matrices \cite{Zhang2014Matrix}:
\begin{align}\label{eq:tool1}
\int_{U(d_1)}...\int_{U(d_\xi)}(U_1\otimes...\otimes U_\xi)^\dagger X(U_1\otimes...\otimes U_\xi)d\mu_H(U_1)...d\mu_H(U_\xi)=\Tr(X)\bigotimes_{j=1}^\xi\frac{I_j}{d_j},
\end{align}
where $d_j$ is the dimension of the unitary matrix $U_j$ and $I_j$ is the identity of $d_j$ dimensions. Therefore, the expectations of the gradients in Eq.~\eqref{eq:S4} can be calculated as
\begin{align}
&\mathbb{E}_{E_i^j\in C}[\partial_{\theta_l}L(\Theta,E_i;\bm{\theta}_0)]=\int_{U(2^m)}...\int_{U(2^m)}d\mu_H(E^1)...d\mu_H(E^\xi)\partial_{\theta_l}L(\Theta,E=\otimes_{j=1}^\xi E^j;\bm{\theta}_0)=0.
\end{align}

Next, we calculate the variances for the gradients by $\text{Var}_{E_i^j\in C}[\partial_{\theta_l}L(\Theta,E_i;\bm{\theta}_0)]=\mathbb{E}_{E_i^j\in C}[\partial_{\theta_l}L(\Theta,E_i;\bm{\theta}_0)^2]$. We utilize the following subspace Haar measure integral:
\begin{align}\label{eq:tool2}
&\int_{U(d_1)}(U_1\otimes I_{2,...,\xi})^\dagger (A_1\otimes I_{2,...,\xi})(U_1\otimes I_{2,...,\xi})X(U_1\otimes I_{2,...,\xi})^\dagger(B_1\otimes I_{2,...,\xi})(U_1\otimes I_{2,...,\xi})d\mu_H(U_1)\nonumber\\
=&\sum_{\nu,\nu'=1}^{d_2d_3...d_\xi}[\int_{U(d_1)}U_1^\dagger A_1U_1(X_{\nu,\nu'}^1)U_1^\dagger B_1U_1d\mu_H(U_1)]\otimes\ket{\nu}\bra{\nu'}\nonumber\\
=&\frac{1}{d_1(d_1^2-1)}[(d_1\Tr(A_1B_1)-\Tr(A_1)\Tr(B_1))\Tr_1(X)\otimes I_{d_1}+(d_1\Tr(A_1)\Tr(B_1)-\Tr(A_1B_1))X].
\end{align}
In the second line, we decompose the operator $X$ into $\sum_{\nu,\nu'}X_{\nu.\nu'}^1\otimes\ket{\nu}\bra{\nu'}$ using the orthogonal basis $\{\ket{\nu}:\nu=1,2,...,d_2d_3...d_\xi\}$. By repeatedly using the above formula, we can evaluate the Haar measure integral of the form $\int_{U(d_1)}...\int_{U(d_\xi)}(U_1\otimes...\otimes U_\xi)^\dagger (A_1\otimes...\otimes A_\xi)(U_1\otimes...\otimes U_\xi)X(U_1\otimes...\otimes U_\xi)^\dagger(B_1\otimes...\otimes B_\xi)(U_1\otimes...\otimes U_\xi)d\mu_H(U_1)...d\mu_H(U_\xi)$. Now, we are ready to compute the variances for the gradients in the adversarial PVQC $U(\bm{\theta})$. We divide the variance into four terms:
\begin{align}\label{eq:A20}
\text{Var}_{E_i^j\in C}[\partial_{\theta_l}L(\Theta,E_i;\bm{\theta}_0)]=&-\Tr[\int_{U(2^m)}...\int_{U(2^m)}d\mu_H(E^1)...d\mu_H(E^{\xi})\rho E^\dagger U_-^\dagger A_lU_+^\dagger EH_V\rho E^\dagger U_-^\dagger A_lU_+^\dagger EH_V]\\\label{eq:A21}
&-\Tr[\int_{U(2^m)}...\int_{U(2^m)}d\mu_H(E^1)...d\mu_H(E^{\xi})\rho H_VE^\dagger U_+A_lU_-E\rho H_VE^\dagger U_+A_lU_-E]\\\label{eq:A22}
&+\Tr[\int_{U(2^m)}...\int_{U(2^m)}d\mu_H(E^1)...d\mu_H(E^{\xi})\rho E^\dagger U_-^\dagger A_lU_+^\dagger EH_V\rho H_VE^\dagger U_+A_lU_-E]\\\label{eq:A23}
&+\Tr[\int_{U(2^m)}...\int_{U(2^m)}d\mu_H(E^1)...d\mu_H(E^{\xi})\rho H_VE^\dagger U_+A_lU_-E\rho E^\dagger U_-^\dagger A_lU_+^\dagger EH_V],
\end{align}
where $E=E^1\otimes...\otimes E^\xi$, $H_V=V^\dagger HV$, and $\rho=\ket{\psi}_{\text{in}}\bra{\psi}_{\text{in}}$. We calculate each term from Eq.~\eqref{eq:A20} to Eq.~\eqref{eq:A23} using Eq.~\eqref{eq:tool1} and Eq.~\eqref{eq:tool2}:
\begin{align}
&\Tr[\int_{U\left(2^m\right)}...\int_{U\left(2^m\right)}d\mu\left(E^1\right)...d\mu\left(E^{\xi}\right)\rho E^\dagger U_-^\dagger A_lU_+^\dagger EH_V\rho E^\dagger U_-^\dagger A_lU_+^\dagger EH_V]\\\nonumber
=&\frac{1}{\left(2^{2m}-1\right)^\xi}\sum_{k,k'}c_kc_{k'}\sum_{J\subseteq \{1,...,\xi\}}\prod_{j\in J}\left(\Tr\left(A_{l,k}^jA_{l,k'}^j\right)-\frac{1}{2^m}\Tr\left(A_{l,k}^j\right)\Tr\left(A_{l,k'}^j\right)\right)\cdot\\
&\prod_{j\notin J}\left(\Tr\left(A_{l,k}^j\right)\Tr\left(A_{l,k'}^j\right)-\frac{1}{2^m}\Tr\left(A_{l,k}^jA_{l,k'}^j\right)\right)\Tr[\Tr_{j\in J}\left(H_V\rho\right)\otimes \bigotimes_{j\in J}I_{2^m}\cdot H_V\rho]\\
&\Tr[\int_{U\left(2^m\right)}...\int_{U\left(2^m\right)}d\mu\left(E^1\right)...d\mu\left(E^{\xi}\right)\rho H_VE^\dagger U_+A_lU_-E\rho H_VE^\dagger U_+A_lU_-E]\\\nonumber
=&\frac{1}{\left(2^{2m}-1\right)^\xi}\sum_{k,k'}c_kc_{k'}\sum_{J\subseteq \{1,...,\xi\}}\prod_{j\in J}\left(\Tr\left(A_{l,k}^jA_{l,k'}^j\right)-\frac{1}{2^m}\Tr\left(A_{l,k}^j\right)\Tr\left(A_{l,k'}^j\right)\right)\cdot\\
&\prod_{j\notin J}\left(\Tr\left(A_{l,k}^j\right)\Tr\left(A_{l,k'}^j\right)-\frac{1}{2^m}\Tr\left(A_{l,k}^jA_{l,k'}^j\right)\right)\Tr[\Tr_{j\in J}\left(\rho H_V\right)\otimes \bigotimes_{j\in J}I_{2^m}\cdot\rho H_V]\\
&\Tr[\int_{U\left(2^m\right)}...\int_{U\left(2^m\right)}d\mu\left(E^1\right)...d\mu\left(E^{\xi}\right)\rho E^\dagger U_-^\dagger A_lU_+^\dagger EH_V\rho H_VE^\dagger U_+A_lU_-E]\\\nonumber
=&\frac{1}{\left(2^{2m}-1\right)^\xi}\sum_{k,k'}c_kc_{k'}\sum_{J\subseteq \{1,...,\xi\}}\prod_{j\in J}\left(\Tr\left(A_{l,k}^jA_{l,k'}^j\right)-\frac{1}{2^m}\Tr\left(A_{l,k}^j\right)\Tr\left(A_{l,k'}^j\right)\right)\cdot\\
&\prod_{j\notin J}\left(\Tr\left(A_{l,k}^j\right)\Tr\left(A_{l,k'}^j\right)-\frac{1}{2^m}\Tr\left(A_{l,k}^jA_{l,k'}^j\right)\right)\Tr[\Tr_{j\in J}\left(H_V\rho H_V\right)\otimes \bigotimes_{j\in J}I_{2^m}\cdot\rho]\\
&\Tr[\int_{U\left(2^m\right)}...\int_{U\left(2^m\right)}d\mu\left(E^1\right)...d\mu\left(E^{\xi}\right)\rho H_VE^\dagger U_+A_lU_-E\rho E^\dagger U_-^\dagger A_lU_+^\dagger EH_V]\\\nonumber
=&\frac{1}{\left(2^{2m}-1\right)^\xi}\sum_{k,k'}c_kc_{k'}\sum_{J\subseteq \{1,...,\xi\}}\prod_{j\in J}\left(\Tr\left(A_{l,k}^jA_{l,k'}^j\right)-\frac{1}{2^m}\Tr\left(A_{l,k}^j\right)\Tr\left(A_{l,k'}^j\right)\right)\cdot\\
&\prod_{j\notin J}\left(\Tr\left(A_{l,k}^j\right)\Tr\left(A_{l,k'}^j\right)-\frac{1}{2^m}\Tr\left(A_{l,k}^jA_{l,k'}^j\right)\right)\Tr[\Tr_{j\in J}\left(\rho\right)\otimes \bigotimes_{j\in J}I_{2^m}\cdot H_V\rho H_V],
\end{align}
where the $\sum_{J\subseteq \{1,...,\xi\}}$ sums over all possible subsets of $\{1,...,\xi\}$ and the operator $A_l$ in each layer of the adversarial PVQC is decomposed as $A_l=\sum_k c_k \bigotimes_{j=1}^\xi A_{l,k}^j$ (Eq.~\eqref{eq:DecompA} of the main text). By summing up these four terms, we obtain the variance of the gradient as: 
\begin{align}
\nonumber&\text{Var}_{E_i^j\in C}[\partial_{\theta_l}L(\Theta,E_i;\bm{\theta}_0)]=\left(\frac{1}{2^{2m}-1}\right)^\xi\sum_{k,k'}c_kc_{k'}\sum_{J\subseteq \{1,...,\xi\}}\prod_{j\in J}\left(\Tr\left(A_{l,k}^jA_{l,k'}^j\right)-\frac{1}{2^m}\Tr\left(A_{l,k}^j\right)\Tr\left(A_{l,k'}^j\right)\right)\cdot\\\nonumber
&\prod_{j\notin J}\left(\Tr\left(A_{l,k}^j\right)\Tr\left(A_{l,k'}^j\right)-\frac{1}{2^m}\Tr\left(A_{l,k}^jA_{l,k'}^j\right)\right)\text{Tr}[\Tr_{j\in J}\left(\rho\right)\otimes\bigotimes_{j\in J}I_{2^m}\cdot H_V\rho H_V+\Tr_{j\in J}\left(H_V\rho H_V\right)\otimes\\
& \bigotimes_{j\in J}I_{2^m}\cdot\rho-\Tr_{j\in J}\left(\rho H_V\right)\otimes \bigotimes_{j\in J}I_{2^m}\cdot\rho H_V-\Tr_{j\in J}\left(H_V\rho\right)\otimes \bigotimes_{j\in J}I_{2^m}\cdot H_V\rho].
\end{align}
We bound the following term with a constant $C_0$, which does not increase with the total system dimension:
\begin{align}
\nonumber C_0=&\max_{J\subseteq\{1,...,\xi\}}\text{Tr}[\Tr_{j\in J}(\rho)\otimes \bigotimes_{j\in J}I_{2^m}\cdot H_V\rho H_V+\Tr_{j\in J}(H_V\rho H_V)\otimes \bigotimes_{j\in J}I_{2^m}\cdot\rho\\
&-\Tr_{j\in J}(\rho H_V)\otimes \bigotimes_{j\in J}I_{2^m}\cdot\rho H_V-\Tr_{j\in J}(H_V\rho)\otimes \bigotimes_{j\in J}I_{2^m}\cdot H_V\rho].
\end{align}

According to the assumption of \cref{thm:2}, $A_{l,k}^j$ is traceless and $\Tr(A_{l,k}^{j2})\leq 2^m,\forall l,i,k$. We have
\begin{align}
&\Tr(A_{l,k}^jA_{l,k'}^j)\leq\sqrt{\Tr(A_{l,k}^{i2})\Tr(A_{l,k'}^{i2})}\leq 2^m,\\
&\abs{\Tr(A_{l,k}^j)\Tr(A_{l,k'}^j)-\frac{1}{2^m}\Tr(A_{l,k}^jA_{l,k'}^j)}\leq 1.
\end{align}

Hence, we bound the variance as
\begin{align}
\text{Var}_{E_i^j\in C}[\partial_{\theta_l}L(\Theta,E_i;\bm{\theta}_0)]&\leq\left(\frac{1}{2^{2m}-1}\right)^\xi\sum_{k,k'}c_kc_{k'}\sum_{J\subseteq \{1,...,\xi\}}\prod_{j\in J} 2^m C_0\\\label{eq:AppA27}
&=\sum_{k,k'}c_kc_{k'}\left(\frac{2^m+1}{2^{2m}-1}\right)^\xi C_0,
\end{align}
which finishes the proof for \cref{thm:2}. 

\section{More Numerical Results}\label{appsec:2}
\begin{figure}
    \centering
    \includegraphics[width=0.99\textwidth]{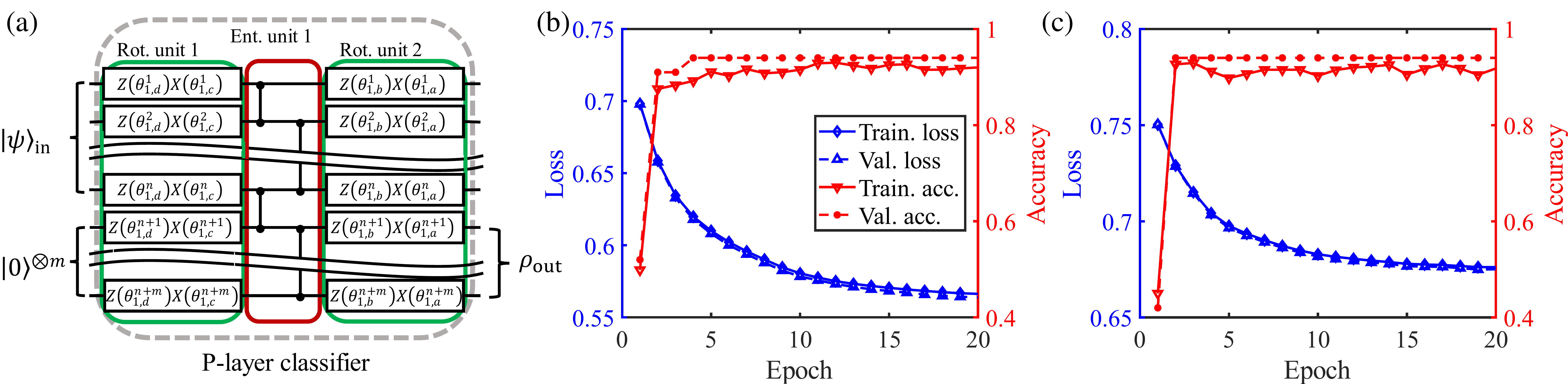}
    \caption{(a) The illustration for $P$-layer variational quantum circuits to construct encoders and classifiers in the numerical simulations. Each layer contains two single-qubit rotation units (the green boxes) and one entangling unit (the red box). In each rotation unit, we perform and Euler angle rotation $Z(\theta_{i,u}^k)X(\theta_{i,v}^k)$. $(u,v)=(d,c)$ or $(b,a)$ distinguishes the two rotation units. $i=1,2,...,P$ denotes the index of the layer. $k=1,2,...,m+n$ denotes the qubit index. (b) and (c). The loss and accuracy averaged over the training set and the validation set as the function of epoch, for the $12$-qubit quantum classifier with (b) the KL divergence and (c) the normalized square loss as the loss function. Each epoch consists of 10 iterations.}
    \label{fig:Apptrain}
\end{figure}

In this appendix, we provide the details for our numerical simulations. The structure of PVQCs we used in numerical simulations is shown in Fig.~\ref{fig:Apptrain}(a). In this $P$-layer PVQC classifier, we first prepare the input state as an $(m+n)$-qubit state $\ket{\psi}_{\text{in}}\otimes\ket{0}^{\otimes m}$, where $\ket{\psi}_{\text{in}}$ is the quantum state that encodes the data to be classified and $\ket{0}^{\otimes m}$ are the ancillary qubits for measurement outputs. Then we apply $P$ layers of unitary quantum operations with each layer containing two rotation units and one entangling unit. Each rotation unit performs an Euler rotation in the single-qubit Bloch sphere and each entangling unit entangles different qubits using CNOT gates between each pair of neighboring qubits. We can adjust the rotation angles and these angles are collectively regarded as variational parameters $\Theta$. The final output state can be written as:
\begin{align}\label{eq:VQCstate}
\ket{\psi(\Theta)}=\left(\prod_{i=1}^PU_i\right)\ket{\psi}_{\text{in}}\otimes\ket{0}^{\otimes m},
\end{align}
where $U_i=\left[\prod_{j=1}^{n+m}Z(\theta_{i,d}^j)X(\theta_{i,c}^j)\right]U_{\text{ent}}\left[\prod_{j=1}^{n+m}Z(\theta_{i,b}^j)X(\theta_{a,1}^j)\right]$ denotes the quantum operation for the $i$-th layer and $U_{\text{ent}}$ denotes the entangling unit. For adversarial attacks and encoders, we set $m=0$. For the quantum classifier, we employ $P=10$ for different system sizes. As the ground states of the cluster Ising model has two phases, we set $m=1$. After the classifier, we measure the output ancillary qubit $\rho_{\text{out}}$ and compute $\Pr(y=m)=\Tr\left(\rho_{\text{out}}\ket{m}\bra{m}\right),m=0,1$. We assign a label $y=0$ if $\Pr(y=0)\geq\Pr(y=1)$ and $y=1$ otherwise.

For the adversarial PVQC $U(\bm{\theta})$, we require that the initial parameters satisfy $U(\bm{\theta}_0)=I$. Therefore, we employ an alternative version of variational quantum circuits. To guarantee that $U(\bm{\theta}_0)=I$, we add a complex conjugate entangling unit at the end of each layer, and each $U_i$ in Eq.~\eqref{eq:VQCstate} becomes $U_i=\left[\prod_{j=1}^{n+m}Z(\theta_{i,d}^j)X(\theta_{i,c}^j)\right]U_{\text{ent}}\left[\prod_{j=1}^{n+m}Z(\theta_{i,b}^j)X(\theta_{a,1}^j)\right]U_{\text{ent}}^\dagger$. We then initialize the adversarial PQVC with all rotation angles being zero.

In the numerical simulations, we exploit the quantum-adapted KL divergence $L_{KL}(h(\ket{\psi}_{\text{in}};\Theta),\bm{p})$ and the normalized square loss $L_{NS}(h(\ket{\psi}_{\text{in}};\Theta),\ket{y})$:
\begin{align}
L_{KL}(h(\ket{\psi}_{\text{in}};\Theta),\textbf{p})&=-\sum_{i=1}^2p_k\log q_k\\
L_{NS}(h(\ket{\psi}_{\text{in}};\Theta),\ket{y})&=1-\abs{\bra{y}\rho_{\text{out}}\ket{y}}^2,
\end{align}
where $\textbf{q}=(q_1,q_2)$ denotes the diagonal elements of the output state $\rho_{\text{out}}$ and $\textbf{p}=(1,0),(0,1)$ for $y=0,1$. During the training procedure of the quantum classifier, we exploit the gradient-based Adam optimization algorithm \cite{Kingma2014Adam,Sashank2018Convergence} to minimize the empirical loss function $L_N(\Theta)=\frac{1}{N}\sum_{i=1}^NL(h(\ket{\psi};\Theta),y)$ over $N$ training samples. To calculate the gradients of the loss function, we employ the definition $\partial L_N(\theta)/\partial\theta=\lim_{\epsilon\to0}\frac{1}{2\epsilon}[L_N(\theta+\epsilon)-L_N(\theta-\epsilon)]$ and estimate the value by choosing a small $\epsilon=10^{-10}$. In Fig.~\ref{fig:Apptrain}(b) and (c), we display the averaged loss and accuracy  in the training procedure for the $12$-qubit classifier. The overfitting risk is low \cite{Srivastava2014Dropout} as the loss and accuracy are close for validation data samples and training data samples. The numerical simulations in this paper were implemented based on the Yao.jl extension \cite{Yao}, the Flux.jl \cite{Innes2018Flux} and the Zygote.jl \cite{Zygote} packages using the Julia programming language \cite{Bezanson2017Julia}.

\section{Analytical Derivations for the Adversarial Risk of Local Unitary Attacks}\label{appsec:3}

In this appendix, we give a detailed proof for \cref{thm:3} in the main text. We introduce the concepts of the concentration function and the Levy family, and some basic results regarding the adversarial machine learning. 

\begin{definition}
For a subset $\mathcal{H}'\subseteq\mathcal{H}$, a \textit{$\tau$-extension} for $\mathcal{H}'$ under the distance metric $D$ is defined as $\mathcal{H}'_\tau=\{x\in\mathcal{H}|D(x,\mathcal{H}')\leq\tau\}$. The \textit{concentration function} for a probability measure $\mu$ is defined as $\alpha(\tau)=1-\inf\{\mu(\mathcal{H}_\tau')|\mu(\mathcal{H})\geq\frac12)\}$. A $d$-dimensional space equipped with the distance metric $D$ and probability measure $\mu$ is called an $(l_1,l_2)$-Levy family if 
\begin{align}
\alpha(\tau)=l_1e^{-l_2\tau^2 d}.
\end{align}
\end{definition}

We next introduce the following lemma showing that $\bigotimes_{i=1}^n SU(2)$ equipped with the Haar measure on each $SU(2)$ and the normalized Hamming distance is a Levy family:

\begin{lemma}
(Example 2.6 in Ref. \cite{Giordano2007Some}) Given a probability measure space $(\mathcal{X},\mu)$, we consider the tensor product space $\mathcal{X}^{\otimes n}$ equipped with $\mu^{\otimes n}$ and the normalized Hamming distance $d_{\text{NH}}$. The ensemble $(\mathcal{X}^{\otimes n},\mu^{\otimes n},d_{\text{NH}})$ forms a $(2,\frac{n}{\dim(\mathcal{X})^n})$-Levy family with the concentration function satisfying:
\begin{align}
\alpha(\tau)=2e^{-\tau^2 n}.
\end{align}
\end{lemma}

We refer to Ref. \cite{Giordano2007Some,Talagrand1995Concentration} for the proof of the above lemma. The above lemma implies that the tensor space $\bigotimes_{i=1}^n SU(2)$ with the Haar measure on each qubit $\mu_H^{\otimes n}$ and the normalized Hamming distance forms a $(2,\frac{n}{2^n})$-Levy family. We recap the following lemma regarding the lower bound of the measure of a $\tau$-extension for a subspace $\mathcal{H}'$.

\begin{lemma}
(Theorem 3.6 in Ref. \cite{Mahloujifar2019Curse}) For a subspace $\mathcal{H}'$ chosen from a $(l_1,l_2)$-Levy family $(\mathcal{H},D,\mu),\dim(\mathcal{H})=d$, the measure of a $\tau$-extension of $\mathcal{H}'$ is greater than $R$ if $\tau$ satisfies
\begin{align}
\tau^2\geq\frac{1}{l_2 d}\ln\left[\frac{l_1^2}{\mu(\mathcal{H}')(1-R)}\right].
\end{align}
\end{lemma}

\begin{proof}
We briefly recap the proof given in Ref. \cite{Mahloujifar2019Curse,Liu2019Vulnerability}. We decompose $\tau$ into two parts $\tau=\tau_1+\tau_2$. We choose $\tau_1$ with $\mu(\mathcal{H}')>l_1e^{-l_2\tau_1^2 d}$. There are two cases concerning whether $\mu(\mathcal{H}')\leq\frac12$:

(1) For the case when $\mu(\mathcal{H}')\leq\frac12$, assuming $\mu(\mathcal{H}'_{\tau_1})\leq\frac12$, we have $\mu(\mathcal{H}\backslash\mathcal{H}'_{\tau_1})\geq\frac12$. For simplicity, we denote $\overline{\mathcal{H}}'_{\tau_1}=\mathcal{H}\backslash\mathcal{H}'_{\tau_1}$. By the definition of Levy family, we deduce that $\alpha(\tau_1)\geq1-\mu(\overline{\mathcal{H}}_{\tau_1}')=\mu(\mathcal{H}')>\alpha(\tau_1)$, which leads to a contradiction.  

(2) For the case when $\mu(\mathcal{H}')>\frac12$, it is straightforward to see $\mu(\mathcal{H}'_{\tau_1})>\frac12$.

Therefore, by choosing such $\tau_1$ we can guarantee that $\mu(\mathcal{H}'_{\tau_1})>\frac12$. Next, we consider the $\tau_2$-extension of $\mathcal{H}'_{\tau_1}$ with $\tau_2^2\geq\frac{1}{l_2d}\ln\left[\frac{l_1}{(1-R)}\right]$. Applying the definition of the Levy family we can prove the lemma as $\mu(\mathcal{H}'_{\tau_1+\tau_2})>1-\alpha(\tau_2)\geq R$ and $\tau^2\geq\tau_1^2+\tau_2^2=\frac{1}{l_2d}\ln\left[\frac{l_1^2}{\mu(\mathcal{H}')(1-R)}\right]$.
\end{proof}

Now we start to prove \cref{thm:3} in the main text. We notice that $(\bigotimes_{i=1}^n SU(2),d_{\text{NH}},\mu_H^{\otimes n})$ forms a $(2,\frac{n}{2^n})$-Levy family and $\dim((\bigotimes_{i=1}^n SU(2))=2^n$. Hence, for any subspace $\mathcal{H}'\subseteq \bigotimes_{i=1}^n SU(2)$, any $\tau$-extension of $\mathcal{H}'$ has measure at least $R^*$ if 
\begin{align}\label{eq:C4}
\tau^2\geq\frac{1}{n}\ln[\frac{4}{\mu_H^{\otimes n}(\mathcal{H}')(1-R^*)}].
\end{align}

Given $k=2,3,...,K$, any data sample in the intersection of the $\tau$-extensions $\{h^{-1}(s_i)_\tau\}_{i=1}^k$ can be transformed into a data sample in any $h^{-1}(s_i)$, when a perturbation $\rho\to\rho'$ of amplitude $\tau$ occurs. The adversarial attack can thus change the labels for all the data samples in this intersection set. By the De Morgan's law, the measure of this intersection set satisfies
\begin{align}
\mu_H^{\otimes n}\left(\cap_{i=1}^kh^{-1}(s_i)_\tau\right)\geq1-\sum_{i=1}^k\mu_H^{\otimes n}\left(\mathcal{H}\backslash h^{-1}(s_i)_\tau\right)=\sum_{i=1}^k\mu_H^{\otimes n}\left(h^{-1}(s_i)_\tau\right)-(k-1).
\end{align}
After setting $R^*=\frac{k-1+R}{k}$ and $\mathcal{H}'=h^{-1}(s_k)$ in Eq.~\eqref{eq:C4}, we deduce that adversarial risk is bounded below by $\mu_H^{\otimes n}\left(\cap_{i=1}^kh^{-1}(s_i)_\tau\right)\geq R$. By choosing the minimal value of all $k=2,3,...,K$, we finish the proof for \cref{thm:3} in the main text.

\end{document}